%% file: paper.tex
\documentclass{llncs}
\usepackage{listings}

\usepackage{xcolor,colortbl}
\usepackage{amssymb}
\usepackage{wasysym}
\usepackage{arydshln} %for nice tables
\usepackage{mathrsfs}
\usepackage{enumerate}
\usepackage[english]{babel}
\usepackage{paralist}
\usepackage{wrapfig}
\usepackage{graphicx}
\usepackage{multirow}

\definecolor{Gray}{gray}{0.85}

\newtheorem{conclusion}{Conclusion}

\newcommand{\texorpdfstring}[2]{#1}

\begin{document}

\title{Shape and Content
\thanks{Kotek, Veith and Zuleger were supported by the Austrian National Research Network S11403-N23 (RiSE) of the Austrian Science Fund (FWF) and by the Vienna Science and Technology Fund (WWTF) through grants PROSEED and ICT12-059. Simkus was supported by the FWF grant P25518 and the WWTF grant ICT12-15}
}
\subtitle{Incorporating Domain Knowledge into Shape Analysis}

\author{D. Calvanese\inst{1} \and T. Kotek\inst{2} \and M. \v{S}imkus\inst{2}
\and H. Veith\inst{2} \and F. Zuleger\inst{2}}

\institute{%Faculty of Computer Science\\
Free University of Bozen-Bolzano
%\\\email{calvanese@inf.unibz.it}
\and
%Faculty of Computer Science\\
Vienna University of Technology\\
%\{\email{kotek,veith,zuleger\}@forsyte.at\\simkus@dbai.tuwien.ac.at}
}

\maketitle

% \documentclass[conference]{IEEEtran}
% \usepackage{listings}
%
%
% \usepackage{amsmath}
% \usepackage{amssymb}
% \usepackage{wasysym}
% \usepackage{arydshln} %for nice tables
% \usepackage{mathrsfs}
% \usepackage{enumerate}
% \usepackage[english]{babel}
% %\usepackage{paralist}
% \usepackage{wrapfig}
% \usepackage{graphicx}
% \usepackage{multirow}
%
% \usepackage{hyperref}
%
%
% \newtheorem{conclusion}{Conclusion}
% \newtheorem{definition}{Definition}
% \newtheorem{theorem}{Theorem}
% \newtheorem{lemma}{Lemma}
% \newtheorem{remark}{Remark}
%
%
% \newcommand{\num}[1]{${}_{#1}$}
% \newcommand{\warn}[1]{\PackageWarning{OurWarnings}{DSL warning>>>> #1}}
%
% %\renewcommand{\texorpdfstring}[2]{#1}
%
% \begin{document}
%
% \title{Shape and Content: Incorporating Domain Knowledge into Shape Analysis
% %\thanks{The first, fourth and fifth authors would like to thank the Austrian National Research
% %Network S11403-N23 (RiSE) of the Austrian Science Fund (FWF) and
% %the Vienna Science and Technology Fund (WWTF) for their support through
% %grants PROSEED, ICT12-059, and VRG11-005.\warn{Mantas and Diego thanks}}
% }
%
% \author{\IEEEauthorblockN{Diego Calvanese}
% \IEEEauthorblockA{Free University of Bozen-Bolzano\\
% Email: calvanese@inf.unibz.it
% }\and
% \IEEEauthorblockN{T. Kotek, M. \v{S}imkus, H. Veith and F. Zuleger}
% \IEEEauthorblockA{Vienna University of Technology\\
% Email: \{kotek,veith,zuleger\}@forsyte.at, simkus@dbai.tuwien.ac.at}
% }
%
%
%
% \maketitle

\begin{abstract}
The verification community has studied dynamic data structures primarily in a bottom-up way by analyzing pointers
and the shapes induced by them.
Recent work in fields such as separation logic has made significant progress in extracting shapes from program source code.
Many real world programs
however manipulate complex data whose structure and content is most naturally described by formalisms
from object oriented programming and databases.
In this paper, we look at the verification of programs with dynamic data structures from the perspective
of content representation.
Our approach is based on description logic, a widely used knowledge
representation paradigm which gives a logical underpinning for diverse modeling frameworks
such as UML and ER.
Technically, we assume that we have separation logic shape invariants
obtained from a shape analysis tool, and requirements on the program data in terms of description logic.
We show that the two-variable fragment of first order logic with counting and trees %(whose decidability was proved at LICS 2013) 
can be used as a joint framework to embed suitable fragments of description logic and
separation logic.
\end{abstract}

%00000000000000000000000000000000000000
%00000000000000000000000000000000000000
%00000000000000000000000000000000000000
%00000000000000000000000000000000000000

\global\long\def\SL{\mathrm{\mathbf{SL}}}
\global\long\def\SLls{\mathrm{\mathbf{SL{\scriptstyle ls}}}}
%\global\long\def\SLind{\mathrm{\mathbf{SL{\scriptstyle ind}}}}
\global\long\def\nil{\mathrm{null}}
\global\long\def\tree{\mathrm{tree}}
\global\long\def\lss{\mathrm{list-of-lists}}
\global\long\def\ls{\mathrm{ls}}
%\global\long\def\iffdef{\stackrel{\mathrm{{\scriptstyle def}}}{\mathrm{iff}}}
\global\long\def\iffdef{\mathrm{iff}}
\global\long\def\arrowfin{\stackrel{\mathrm{{\scriptscriptstyle fin}}}{\mathrm{\to}}}
\global\long\def\false{\mathbf{F}}
\global\long\def\true{\mathbf{T}}
\global\long\def\emp{\mathrm{emp}}
\global\long\def\NR{\mathsf{N_{R}}}
\global\long\def\NRF{\mathsf{N_{F}}}
\global\long\def\NC{\mathsf{N_{C}}}
\global\long\def\NCshp{\mathsf{N_{C}^{shp}}}
\global\long\def\NCdom{\mathsf{N_{C}^{dom}}}
\global\long\def\NI{\mathsf{N_{I}}}
\global\long\def\NIdom{\mathsf{N_{I}^{dom}}}
\global\long\def\NIvar{\mathsf{N_{I}^{var}}}
\global\long\def\NIloc{\mathsf{N_{I}^{free}}}
\global\long\def\lb{[\![}
\global\long\def\rb{]\!]}
\global\long\def\bv{\,\brokenvert\:}
\global\long\def\mm{\mathcal{M}}
\global\long\def\nn{\mathcal{N}}
\global\long\def\abo{\mbox{abort}}
\global\long\def\mmdom{M}
\global\long\def\sM{\mathscr{s}(\mm)}
\global\long\def\tm{\tau_{m}}
\global\long\def\tauField{\tau_{\mathrm{fields}}}
%\global\long\def\tauFree{\tau_{\mathrm{free}}}
\global\long\def\tauVar{\tau_{\mathrm{var}}}
\global\long\def\eqqm{\stackrel{{\scriptscriptstyle ?}}{=}}
\global\long\def\bigsqcap{\sqcap}
\global\long\def\L{\mathcal{L}}
\global\long\def\DL{\mathcal{ALCHOIQ}br}
\global\long\def\err{\mbox{err}}
\global\long\def\FO{\mathrm{FO}}

\newcommand{\Addr}{Addresses}
\newcommand{\nncur}{\nn_{cur}}
\newcommand{\nngho}{\nn_{gho}}
\newcommand{\mmcur}{\mm_{cur}}
\newcommand{\mmgho}{\mm_{gho}}
\newcommand{\gho}{{gho}}
\newcommand{\MemPool}{MemPool}
\newcommand{\annot}{{annot}}
\newcommand{\sannot}{{shp}}
\newcommand{\vercond}{{vercond}}
\newcommand{\control}{{ctrl}}
\newcommand{\pp}{\mathcal{P}}
\renewcommand{\ll}{LL}
\newcommand{\pre}{{pre}}
\newcommand{\post}{{post}}
\newcommand{\LLPrograms}{LLPrograms}
\newcommand{\linit}{\ell_{init}}
\newcommand{\Init}{Init}
\newcommand{\ctwotree}{CT^2}
\newcommand{\ctwo}{C^2}
\newcommand{\Fone}{F_1}
\newcommand{\Ftwo}{F_2}
\newcommand{\Fnosub}{F}

\newif\ifCont
\Conttrue

\ifCont
\newcommand{\dmnCnt}{cnt}
\newcommand{\domainContent}{content}
\newcommand{\DomainContent}{Content}
\else
\newcommand{\dmnCnt}{dmn}
\newcommand{\domainContent}{domain}
\newcommand{\DomainContent}{Domain}
\fi

\newif\ifSkip
\Skipfalse

\newif\ifiFM
\iFMfalse

\pagestyle{plain}

\input{introduction}

\section{Logics for Invariant Specification} \label{se:logics}

\subsection{Memory Structures \label{subse:structures--1}\label{subse:memory-structures--1}}

We use ordinary first order structures to represent memory in a
precise way. A \emph{structure} (or, \emph{interpretation}) is a tuple
$\mm=(M,\tau,\cdot)$, where
(i) $M$ is an infinite set (the \emph{universe}),
(ii) $\tau$ is a set of \emph{constants} and \emph{relation symbols}
  with an associated non-negative arity, and
(iii) $\cdot$ is an \emph{interpretation function}, which assigns to
  each constant $c\in \tau$ an element $c^\mm\in M$, and to each
  $n$-ary relation symbol $R\in \tau$ an $n$-ary relation $R^{\mm}$
  over $M$. Each relation is either unary or binary
  (i.e. $n\in\{1,2\}$).
  Given $A\subseteq M$, a binary $R^\mm$, $R^\mm$ and $e\in A^\mm$, we may use the notation
  $R^\mm(e)$ if $R^\mm$ is known to be a function over $A^\mm$.

A {\em Memory structure} describes a snapshot of the heap and the
local variables. We assume sets
$\tauVar\subseteq \tau$ of constants $\tauField\subseteq \tau$ of
binary relation symbols.  We will later employ these symbols for
variables and fields in programs.  A {\em memory structure} is a structure
$\mm=(M,\tau,\cdot)$ that satisfies the following conditions:
\begin{enumerate}[(1)]
\item % $\tauVar\cup \tauField\subseteq \tau$ and
  $\tau$ includes the constants $o_{\nil}$, $o_\true$, $o_\false$.
\item $\tau$ has the unary relations $\Addr$, $Alloc$, $PossibleTargets$, $\MemPool$,
  and $Aux$.
\item $Aux^\mm=\{o_\nil^\mm,o_\true^\mm,o_\false^\mm\}$ and $|Aux^\mm|=3$.
\item $\Addr^\mm\cap Aux^\mm= \emptyset$ and $\Addr^\mm \cup Aux^\mm =
  M$.
\item $Alloc^\mm$, $PossibleTargets^\mm$ and $\MemPool^\mm$ form a partition of $\Addr^\mm$.
\item $c^\mm\in M\backslash \MemPool^\mm$ for every constant $c$ of $\tau$.
\item For all $f\in  \tauField$, $f^\mm$ is a function from
  $\Addr^\mm$ to $M\backslash \MemPool^\mm$.
\item If $e\in \MemPool^\mm$, then $f^\mm(e)\in \{o_\nil^\mm,o_\false^\mm\}$.
\item $R^{\mm}\subseteq (M\backslash \MemPool^\mm)^n$ for every\footnote{Here $n\in \{1,2\}$.}
  $n$-ary $R\in \tau \setminus (\{\MemPool \}\cup \tauField)$.
 \item $Alloc^\mm$ and $PossibleTargets^\mm$  are finite. $\MemPool^\mm$ is infinite.

\end{enumerate}

We explain the intuition behind memory structures. Variables in
programs will either have a Boolean value or be pointers. Thus, to
represent $\nil$ and the Boolean values $\true$ and $\false$, we
employ the auxiliary relation $Aux^\mm$ storing 3 elements
corresponding to the 3 values. $\Addr^{\mm}$ represents the memory
cells.
The relation $Alloc^\mm$ is the set of allocated
cells, $PossibleTargets^\mm$ contains all cells which are not allocated,
but are pointed to by allocated cells (for technical reasons it possibly contains some other unallocated cells).
$\MemPool^\mm$ contains the cells which are not allocated, do
not have any field values other than $\nil$ and $\false$, are not
pointed to by any field, do not participate in any other relation and
do not interpret any constant (see (6-9)). The memory cells in $\MemPool$ are
the candidates for allocation during the run of a program.
 Since the allocated memory should by finite at any point of the
execution of a program, we require that
$Alloc^\mm $ and $PossibleTargets^\mm$
are finite (see (10)), while the available memory $\Addr^\mm$ and the memory pool $\MemPool^\mm $
are infinite. Finally, each
cell is seen as a record with the fields of
$\tauField$.

\subsection{The Description Logic
 $\L$}

$\L$ is defined w.r.t. a vocabulary $\tau$ consisting of
relation and constant symbols.  \footnote{In DL terms, $\L$
  corresponds to Boolean $\mathcal{ALCHOIF}$ knowledge bases with the
  additional support for role intersection, role union, role
  difference and product roles.  }
% and complex formulae built from them.
% The semantics to formulae
% is given in terms of structures, where atomic concepts and atomic
% roles are interpreted as unary and binary relations in a structure,
% respectively, and constants are interpreted as elements in the
% structure's universe.

\begin{definition}[Syntax of $\L$]
  The sets of \emph{roles} and  \emph{concepts} of $\L$ is defined
  inductively:
  (1) every unary relation symbol  is a concept (\emph{atomic concept});
  (2)  every constant symbol is a concept;
  (3)  every binary relation symbol is a role (\emph{atomic role});
  (4) if $r,s$ are roles, then $r\cup s$, $r\cap s$,
    $r\backslash s$ and $r^{-}$ are roles;
 (5)  if $C,D$ are concepts, then so are
    $C\sqcap D$, $C\sqcup D$, and $\neg C$;
  % \item if $r$ is a role, $C$ is a concept, and $n$ is a non-negative
  %   integer, then $\exists^{\geq n} r.C$ and $ \exists^{\leq n} r.C$ are also concepts;
  (6) if $r$ is a role and $C$ is a concept, then $\exists r.C$ is
    also a  concept;
(7)  if $C,D$ are concepts, then $C\times D$ is a role
(\emph{product role}).

  The set of \emph{formulae} of $\L$ is the closure under $\land$,$\lor$,$\neg$,$\to$ of the
  atomic formulae:
  $C\sqsubseteq D$ (\emph{concept inclusion}), where $C,D$
    are concepts;
  $r\sqsubseteq s$ (\emph{role inclusion}), where $r,s$
    are roles; and
  $func(r)$ (\emph{functionality assertion}), where $r$
    is a role.

\end{definition}

 \begin{definition}[Semantics of $\L$]
   The semantics is given in terms of
   structures $\mm=(M,\tau,\cdot)$.
   % where $\tau$ contains all atomic
   %concepts, atomic roles and constants occurring in $\varphi$.
   The extension of $\cdot^{\mm}$ from the atomic relations and constants in $\mm$
    and the
   satisfaction relation $\models$ are
   given below. If $\mm\models \varphi$,
   then $\mm$ is a \emph{model} of $\varphi$. We write $\psi \models
   \varphi$ if every model of $\psi$ is also a model of $\varphi$.
\[
\begin{array}{l}
\begin{array}{llllll}
  (C\sqcap  D)^{\mm} & = & C^{\mm}\cap D^{\mm} \hphantom{M\setminus \ \ \ }& (r\sqcap  s)^{\mm} &=& r^{\mm}\cap s^{\mm} \\

  (C\sqcup  D)^{\mm} & = & C^{\mm}\cup D^{\mm} \hphantom{M\setminus \ \ \ }&  (r\sqcup  s)^{\mm}  &=&  r^{\mm}\cup s^{\mm} \\

  (\neg C )^{\mm} & = & M\setminus C^{\mm} \hphantom{C^{\mm}\cap \ \ \ }& (r\setminus  s)^{\mm} &=&  r^{\mm}\setminus s^{\mm} \\

 (C\times D)^{\mm} & = & C^{\mm}\times D^{\mm} & (r^-)^{\mm} &=&  \{(e,e')\mid (e',e)\in r^{\mm} )\}\\

(\exists r.C)^{\mm} & = & \{e\mid \exists e' : (e,e')\in r^{\mm}\}\ \ \ \\

\mm\models C\sqsubseteq D & \mathrm{if}  & C^\mm\subseteq D^\mm  &
\mm\models r\sqsubseteq s& \mathrm{if}  & r^\mm\subseteq s^\mm \\
\end{array}\\
\begin{array}{llllll}
\mm\models func(r) & \mathrm{if}   \{(e,e_1),(e,e_2)\}\subseteq r^\mm~\mathrm{implies}~e_1=e_2\\
\end{array}
\end{array}
\]
\end{definition}
The closure of $\models$ under $\land$ $\lor$,$\neg$,$\to$ is defined in the natural way.
We abbreviate:\\
$\top=C\sqcup \neg C$, where $C$ is an arbitrary atomic concept and $\bot=\neg \top$;
$\alpha \equiv \beta$ for the formula $\alpha\sqsubseteq
  \beta\land \beta\sqsubseteq \alpha$; and
 $\exists r$ for  the concept  $\exists r.\top$;
$(o,o')$ for the role $o\times o'$.
Note that $\top^\mm=M$ and $\bot^\mm=\emptyset$
  for any structure $\mm=(M,\tau,\cdot)$.

%
% \begin{definition}[Satisfiability and implication in memory structures]
%   An $\L$-formula $\varphi$ is \emph{satisfiable in a memory
%     structure} if there is a memory structure $\mm$ such that
%   $\mm\models \varphi$. We write $\psi \models_m \varphi$ if
%   $\mm\models \psi$ implies $\mm\models \varphi$ for every memory
%   structure $\mm$.
% \end{definition}
%
%
% The following is an easy consequence of the existing results in the
% DL literature.
% \begin{theorem}\label{thm:nexptime}Checking whether an $\L$-formula is satisfiable in a
%   memory structure is $\mathrm{NExpTime}$-complete.
% \end{theorem}
% \begin{IEEEproof}[Sketch] The upper bound can be shown by reducing the
%   problem to finite satisfiability in $\mathcal{C}^2$, the
%   two-variable fragment of the first-order logic with counting
%   quantifiers~\cite{ar:PrattHartmann2005}. The lower bound comes
%    from the $\mathrm{NExpTime}$-hardness of satisfiability in the
%   DL $\mathcal{ALCHOIQ}$~\cite{DBLP:journals/jair/Tobies00}.
% \end{IEEEproof}
%
% Clearly, since $\psi\models_m \varphi $ iff $\psi\land \neg \varphi$ is
% \emph{not} satisfiable in a memory structure, we get that deciding
% $\psi\models_m \varphi $ for a pair $\psi,\varphi$ of $\L$-formula is
% $\mathrm{coNExpTime}$-complete.
%
%
%
%
% In the rest of the paper, we will only be interested in satisfiability
% and implication of formulae in memory structures. Thus, the term
% \emph{satisfiability} will always refer to \emph{satisfiability in a
%   memory structure}, and we will use $\models$ instead of $\models_m$.

\subsection{Running Example: \DomainContent~Invariants in $\L$} %\label{se:expressing-invariants--1}
\label{se:expressing-content-invariants}
\label{se:EPM-example-DL}
Now we make the example from Section \ref{se:running-intro} more precise.
The concepts $ELst$ and $PLst$ are interpreted as the sets of elements in
the employee list resp. the project list.
$mngBy$, $isMngr$ and $wrkFor$ are roles. $o_{eHd}$ and $o_{pHd}$ are the constants
which correspond to the heads of the two lists.
The invariants of the systems are:\\
\begin{tabular}{lllr}
\arrayrulecolor{Gray} \hline
The emploee and project lists are allocated:\ \ &
$PLst\sqcup ELst$ & $\sqsubseteq$ & $Alloc$\\
\arrayrulecolor{Gray} \hline

Projects and employees are distinct:&
$PLst\sqcap ELst$ & $\sqsubseteq$ & $\bot$ \\
\arrayrulecolor{Gray} \hline

$wrkFor$ is set to null for projects:&
$PLst$ & $\sqsubseteq$ & $\exists wrkFor.o_{\nil}$ \\
\arrayrulecolor{Gray} \hline

$mngBy$ is set to null for employees:&
$ELst$ & $\sqsubseteq$ & $\exists mngBy.o_{\nil}$ \\
\arrayrulecolor{Gray} \hline

$wrkFor$ of employees in the list point\\
to projects in the list or to null:&
$\exists wrkFor^{-}.ELst$ & $\sqsubseteq$ & $PLst\sqcup o_{\nil}$ \\
\arrayrulecolor{Gray} \hline

$isMngr$ is a Boolean field:&
$\exists isMngr^{-}.ELst$ & $\sqsubseteq$ & $Boolean$ \\
\arrayrulecolor{Gray} \hline
\end{tabular}\\
\begin{tabular}{lr}
$mngBy$ of projects point \ \ \ \ \ \ \ \ \ \ \ \ \ \ \ \ \ \ \ \ \ \ \ \ \ & $\exists mngBy^{-}.PLst\sqsubseteq  $
 \\
to managers or null: &
$(ELst\sqcap \exists isMngr.o_{\true}) \sqcup o_\nil$\\
\arrayrulecolor{Gray} \hline

The manager of a project \\

must work for the project:&
$mngBy \cap (\top\times ELst)$
  $\sqsubseteq wrkFor^{-}$\\
  \arrayrulecolor{Gray} \hline

 \end{tabular}

Let the conjunction of the invariants be given by $\varphi_{invariants}$.

Consider $S$ from Section \ref{se:intro}.
The states of the heap before and after the execution of $S$
can be related by the following $\L$ formulae. $\varphi_{lists-updt}$ and $\varphi_{p-assgn}$.
$\varphi_{lists-updt}$ states that the employee list at the end of the program ($ELst$)
is equal to the employee list at the beginning of the program ($ELst_\gho$),
and that the project list at the end of the program ($PLst$)
is the same as the project list at the beginning of the program ($PLst_\gho$),
except that $PLst$ also contains the new project $o_{proj}$.
$ELst_{\gho}$ and $wrkFor_{\gho}$ are {\em ghost relation symbols},
whose interpretations hold the corresponding values at the beginning of $S$.
\begin{eqnarray*}
\varphi_{lists-updt} &=& ELst_{\gho} \equiv ELst \land  PLst_{\gho}\sqcup o_{proj}  \equiv  PLst\\
\varphi_{p-assgn} &=& ELst_{\gho}\sqcap\exists wrkFor_{\gho}.o_{\nil}  \equiv ELst\sqcap\exists wrkFor.o_{proj}
\end{eqnarray*}
\subsubsection{Ghost symbols}
As discussed in Section \ref{se:EPM-example-DL},
in order to allow invariants of the form
$
\begin{array}{lll}
 \varphi_{lists-updt} &=& ELst_{\gho} \equiv ELst  \land
  PLst_{\gho}\sqcup o_{proj}  \equiv  PLst
\end{array}
$\\
we need ghost symbols.
We assume $\tau$ contains, for every symbol e.g. $s \in \tau$, the symbol $s_{\gho}$.
Therefore, memory structures actually contain {\em two}
snapshots of the memory:
one is the {\em current} snapshot, on which the program operates,
and the other is a {\em ghost} snapshot, which is a snapshot of the memory
 at the beginning of the program, and which the program does not change or interact with.
We denote the two underlying memory structures of $\mm$ by
$\mmcur$ and $\mmgho$.
Since the interpretations of ghost symbols should not change throughout the run of a program,
they will sometime require special treatment.

\subsection{The Separation Logic Fragment $\SLls$} \label{se:SLls}

The SL that we use is denoted $\SLls$, and is
the logic from \cite{pr:BCO05} with lists and multiple pointer fields, but without trees.
It can express that the heap is partitioned into lists and individual cells.
For example, to express that the heap contains only the two lists $ELst$ and $PLst$
we can write the $\SLls$ formula $\ls(pHd,\nil)*\ls(eHd,\nil)$.

We denote by $var_{i}\in Var$ and $f_{i}\in Fields$ the sets of variables respectively
fields to be used in $\SLls$-formulae. $var_i$ are constant symbols. $f_i$ are binary relation symbols
always interpreted as functions.
An $\SLls$-formula $\Pi\bv\Sigma$ is the conjunction
of a {\em pure part} $\Pi$ and a {\em spatial part} $\Sigma$.
$\Pi$ is a conjunction of equalities and inequalities of variables and $o_{null}$.
$\Sigma$ is a {\em spatial conjunction} $\Sigma = \beta_1 *\cdots *\beta_r$ of formulae of the form
$\ls(E_1,E_2)$ and $var\mapsto[f_{1}:E_{1},\ldots,f_{k}:E_{k}]$, where
each $E_i$ is a variable or $o_{null}$. 
Additionally, $\Sigma$ can be $emp$ and $\Pi$ can be $\true$. 
When $\Pi=\true$ we write $\Pi\bv\Sigma$ simply as $\Sigma$.

The memory model of \cite{pr:BCO05} is very similar to ours.
We give the semantics of $\SLls$ in memory structures directly due to space constraints. 
\ifiFM 
See the full paper \cite{full-version}
\else
See the appendix
\fi
for a discussion of the standard semantics of $\SLls$.
$\Pi$ is interpreted in the natural way. $\Sigma$ indicates that $Alloc^\mm$
is the disjoint union of $r$ parts $P_1^\mm,\ldots,P_r^\mm$.
If $\beta_i$ is of the form $var\mapsto[f_{1}:E_{1},\ldots,f_{k}:E_{k}]$
then $|P_i^\mm|=1$ and, denoting $v\in P_i^\mm$, $f_j^\mm(v)=E_j^\mm$.
If $\beta_i$ is of the form $\ls(E_{1},E_2)]$, then
$|P_i^\mm|$ is a list from $E_1^\mm$ to $E_2^\mm$. $E_2^\mm$ might not belong to $P_i^\mm$.
If $\Sigma=emp$ then $Alloc^\mm=\emptyset$. 

\subsection{The Two-variable Fragment with Counting and Trees $\ctwotree$}
$\ctwo$ is the subset of first-order logic whose formulae contain at most two variables,
extended with counting quantifiers $\exists^{\leq k}$, $\exists^{\geq k}$ and $\exists^{= k}$ for all $k\in \mathbb{N}$.
W. Charatonik and P. Witkowski \cite{pr:CharatonikWitkowski13} recently studied an extension of $\ctwo$
which trees which, as we will see, contains both our DL and our SL.
$\ctwotree$ is the subset of second-order logic of the form
$
 \exists \Fone \, \varphi(\Fone) \land \varphi_{forest}(\Fone)
$
where $\varphi\in\ctwo$ and $\varphi_{forest}(\Fone)$ says that $\Fone$ is a forest.
Note that $\ctwotree$ is not closed under negation, conjunction or disjunction.
However, $\ctwotree$ is closed under conjunction or disjunction with $\ctwo$-formulae.

 A $\ctwotree$-formula $\varphi$ is \emph{satisfiable in a memory
    structure} if there is a memory structure $\mm$ such that
  $\mm\models \varphi$. We write $\psi \models_m \varphi$ if
  $\mm\models \psi$ implies $\mm\models \varphi$ for every memory
  structure $\mm$.
Lemma \ref{lem:fin-sat-mem}
states the crucial property of $\ctwotree$ that we use.
It follows from \cite{pr:CharatonikWitkowski13}, by reducing
the memory structures to closely related finite structures.
\footnote{In fact \cite{pr:CharatonikWitkowski13} allows existential quantification over two forests, but will only need one. }
\ifiFM
(see full version \cite{full-version}).
\else
(see Appendix \ref{app:nexptime}).
\fi
\begin{lemma}\label{lem:fin-sat-mem}
Satisfiability of $\ctwotree$ by memory structures is in NEXPTIME.
\end{lemma}

\subsection{Embedding $\L$ and $\SLls$ in $\ctwotree$}\label{se:embed}

$\L$ has a fairly standard reduction (see
e.g. \cite{DBLP:journals/ai/Borgida96}) to $\mathcal{C}^2$:
\begin{lemma}\label{lem:dl-to-ct2}
For every vocabulary, there exists
$tr:\L(\tau)\to\ctwo(\tau)$ such that for every $\varphi\in\L(\tau)$,
$\varphi $ and $tr(\varphi)$ agree on the truth value of all $\tau$-structures.
\end{lemma}
E.g., $tr(C_1\sqsubseteq C_2) = \forall x\, C_1(x)\to C_2(x)$. The details of $tr$
are given in Table \ref{tab:fo2}.

The translation of $\SLls$ requires more work. Later we need the following
related translations:
$\alpha:\SLls\to\L$ extracts from the $\SLls$ properties whatever can be expressed in $\L$.
$\beta:\SLls\to\ctwotree$ captures $\SLls$ precisely.

Given a structure $\mm$, $L^\mm$ is a {\em singly linked list from $o_{var_1}^\mm$ to $o_{var_2}^\mm$ w.r.t. the field $next^\mm$}
if $\mm$ satisfies the following five conditions, or it is empty. Except for (5), the conditions are expressed fully in $\L$ below:
\\
(1) $o_{var_1}^\mm$ belongs to $L^\mm$;
(2) $o_{var_2}^\mm$ is pointed to by an $L^\mm$ element;
(3) $o_{var_2}^\mm$ does not belong to $L^\mm$;
(4) Every $L^\mm$ element is pointed to from an $L^\mm$ element,
      except possibly for $o_{var_1}^\mm$;
(5) all elements of $L^\mm$ are reachable from $o_{var_1}^\mm$ via $next^\mm$.
Let
\[
\begin{array}{l}
\begin{array}{llllll}
\alpha^1(\ls)&=&(o_{var_1}  \sqsubseteq  L)  & \alpha^3(\ls)&=&(o_{var_2}  \sqsubseteq  \neg L)\\
\alpha^2(\ls)&=&(o_{var_2}  \sqsubseteq  \exists next^{-}.L) \ \ \ \ \ &
\alpha^4(\ls)&=&(L  \sqsubseteq  o_{var_1}\sqcup \exists next^{-}.L)
\end{array}\\
\begin{array}{ll}
\alpha_{emp-ls}(\ls) &= (L\sqsubseteq \bot) \land (o_{var_1}=o_{var_2})\\
\alpha(\ls) & = \alpha^1(\ls) \land \cdots \land \alpha^4(\ls)\lor \alpha_{emp-ls}(\ls)
\end{array}
\end{array}
\]
In memory structures $\mm$ satisfying $\alpha(\ls)$,
if $L^\mm$ is not empty, then it contains a list segment from  $o_{var_1}^\mm$ to $o_{var_2}^\mm$,
but additionally $L^\mm$ may contain additional simple $next^\mm$-cycles,
which are disjoint from the list segment.
Here we use the finiteness of $Alloc^\mm$ (which contains $L^\mm$) and the functionality of $next^\mm$.
A connectivity condition is all that is lacking to express $\ls$ precisely.
$\alpha(\ls)$ can be extended to $\alpha:\SLls\to\L$ in a natural way 
\ifiFM
(see the full version \cite{full-version})
\else
(see Appendix \ref{app:translations}) 
\fi 
such that:
\begin{lemma}\label{lem:sl-to-dl-2-alpha}
 For every $\varphi\in\SLls$, $\varphi$ implies $\alpha(\varphi)$ over memory structures.
\end{lemma}

To rule out the superfluous cycles we turn to $\ctwotree$.
Let
\ifiFM
$\beta^5(\ls) = \forall x\forall y\, \big[(L(x)\land L(y))\to (\Fone(x,y) \leftrightarrow next(x,y))\big]\land 
  \forall x \big[\big(L(x)\land\forall y \,(L(y)\to\neg \Fone(y,x))\big)\to (x\approx o_{var_1})\big]
$.
\else
\[
\begin{array}{l}
\beta^5(\ls) = \forall x\forall y\, \big[(L(x)\land L(y))\to (\Fone(x,y) \leftrightarrow next(x,y))\big]\land \\
 \hphantom{\beta^5(\ls) = } \forall x \big[\big(L(x)\land\forall y \,(L(y)\to\neg \Fone(y,x))\big)\to (x\approx o_{var_1})\big]
\end{array}
\]
\fi
$\beta^5(\ls)$ states that the forest $F_1$ coincides with $next$ inside $L$ and
that the forest induced by $\Fone$ on $L$ is a tree.
Let
$
\beta(\ls)= \exists \Fone\, tr(\alpha(\ls)) \land \beta^5(\ls) \land \varphi_{forest}(F_1)
$.
$\beta(\ls)\in \ctwotree$ and it expresses that $L^\mm$ is a list.
The extension of $\beta(\ls)$ to the translation function $\beta:\SLls\to\ctwotree$ is natural
and discussed in 
\ifiFM 
the full version \cite{full-version}.
\else Appendix \ref{app:translations}.
\fi
\ifiFM 
The full version \cite{full-version} also
\else
Appendix \ref{app:cyclic}
\fi
discusses the translation of {\em cyclic data structures} under $\beta$. 
\begin{lemma}\label{lem:sl-to-dl-2-beta}
 For every $\varphi\in\SLls$:
$\varphi$ and $\beta(\varphi)$ agree on all memory structures.
\end{lemma}

% %-- gamma ------------------------------------------------------------------------------------------
% In the sequel we will need to be able to also reason over the negations of $\SLls$ formulae.
% However, since $\ctwotree$ is not closed under negation, we need to define this directly.
% We define a new translation function $\gamma:\SLls\to\L$ satisfying:
% \begin{lemma}\label{lem:sl-to-dl-2}
% \label{lem:weak-list} For every heap $h$ and stack $s$, if $\varphi\in\SLls$ then:
% $s,h\not\models\varphi$ iff $\mm_{s,h} \models \gamma(\varphi)$.
% \end{lemma}
% To define $\gamma_\ls$, let
% \[
% \gamma_\ls = \neg \alpha^1(\ls) \lor \cdots \lor \neg \alpha^4(\ls) \lor \gamma_\ls^5
% \]
% where $\gamma_\ls^5$ is defined as follows:
% \[
%  \begin{array}{lll}
%  \gamma_\ls^5 &=&(L\sqcap Y \times L\sqcap \neg Y)\cap next \equiv \bot \\
%  &&\land o_{var_1}\sqsubseteq Y \land \bot \not\equiv \neg Y\sqcap(L\sqcup o_{var_2})
%  \end{array}
% \]
% $Y$ is a fresh concept name. $\gamma_\ls^5$ expresses that there is an element of  $L^\mm \cup {o_{var_2}^\mm}$  which is not reachable via the $next^\mm$ pointer from $o_{var_1}^\mm$.
% $\gamma_\ls$ expresses that $L^\mm$ is not a list from $o_{var_1}^\mm$ to $o_{var_2}^\mm$ via $next^\mm$.
% Again we leave the details of the extension of $\gamma$ to full $\SLls$ to the appendix.
% %--------------------------------------------------------------------------------------------

$\ctwotree$'s flexibility allows to easily express variations
of singly-linked lists, such as doubly-linked lists, or lists in which
every element points to a special head element via a pointer $head$, and analogue variants of trees.

\subsection{Running Example: Shape Invariants}\label{se:shape-inv-example}
At the loop header of the program $S$ from the introduction,
the memory contains two distinct lists, namely $PLst$ and $ELst$.
$ELst$ is partitioned into two parts: the employees who have been visited in the loop so far,
and those that have not.
This can be expressed in $\SLls$ by the formula:
$\varphi_{\ell_l} = \true \bv \ls(eHd,e) * \ls(e,nil) * \ls (pHd, nil)$.
The translation $\alpha(\varphi_{\ell_l})$ is given by
\ifiFM
$
P_1\sqcup P_2\sqcup P_3 \equiv Alloc
  \hphantom{;} \land \nonumber
  \alpha(\ls(eHd,e,next,P_1)) \land  \nonumber 
  \alpha(\ls(e,\nil,next,P_2))\land \nonumber
  \alpha(\ls(pHd,\nil,next,P_3))\land \nonumber 
   P_1 \sqcap P_2 \equiv \bot  \land P_1 \sqcap P_3 \equiv \bot
   \land  P_2 \sqcap P_3\equiv \bot \land \alpha_{\true} \nonumber
$
\else
\[
\begin{array}{l}
  P_1\sqcup P_2\sqcup P_3 \equiv Alloc
  \hphantom{;} \land \nonumber
  \alpha(\ls(eHd,e,next,P_1)) \land  \nonumber \\
  \alpha(\ls(e,\nil,next,P_2))\land \nonumber
  \alpha(\ls(pHd,\nil,next,P_3))\land \nonumber \\
   P_1 \sqcap P_2 \equiv \bot  \land P_1 \sqcap P_3 \equiv \bot
   \land  P_2 \sqcap P_3\equiv \bot \land \alpha_{\true} \nonumber
\end{array}
\]
\fi
The translation from SL assigns concepts $P_i$ to each of
the lists.
$\alpha_{\true}$ which occurs in $\alpha(\varphi_{\ell_l})$ is the translation of $\Pi=\true$
in $\varphi_{\ell_l}$.
In order to clarify the meaning of $\alpha(\varphi_{\ell_l})$ we relate the $P_i$
to the concept names from Section \ref{se:EPM-example-DL} and simplify the formula somewhat.
Let $\psi_l = P_1\sqcup P_2 \equiv ELst \land P_3\equiv PLst$.
$P_1$ contains the
elements of $ELst$ visited in the loop so far.
$\alpha(\varphi_{\ell_l})$ is equivalent to:
\ifiFM
$
\alpha'(\varphi_{\ell_l}) =    \psi_l \land
   ELst \sqcup PLst \equiv Alloc \land
  ELst \sqcap PLst \equiv \bot \land  \alpha(\ls(eHd,e,next,P_1)) 
 \land\alpha(\ls(e,\nil,next,ELst \sqcap \neg P_1))\land
 \alpha(\ls(pHd,\nil,next,PLst))$. We have 
$\beta^5(\Sigma) = \beta^5(\ls(eHd,e,next,P_1))\land
\beta^5(\ls(e,\nil,next,ELst \sqcap \neg P_1)) \land
 \beta^5(\ls(pHd,\nil,next,PLst))$
and
$
\beta(\varphi_{\ell_l}) =\exists \Fone\, tr(\alpha(\varphi) \land \beta^5(\Sigma) \land \varphi_{forest}(F_1)
$.
\else
\[
\begin{array}{llll}
\alpha'(\varphi_{\ell_l}) &=&    \psi_l \land
   ELst \sqcup PLst \equiv Alloc \land
  ELst \sqcap PLst \equiv \bot \land  \alpha(\ls(eHd,e,next,P_1)) \\ &
& \land\alpha(\ls(e,\nil,next,ELst \sqcap \neg P_1))\land
 \alpha(\ls(pHd,\nil,next,PLst))\\
\beta^5(\Sigma) & =& \beta^5(\ls(eHd,e,next,P_1))\land
\beta^5(\ls(e,\nil,next,ELst \sqcap \neg P_1)) \land
\\
&& \beta^5(\ls(pHd,\nil,next,PLst))
\\
\beta(\varphi_{\ell_l}) &=&\exists \Fone\, tr(\alpha(\varphi) \land \beta^5(\Sigma) \land \varphi_{forest}(F_1)
\end{array}
\]
\fi
%%%%%%%%%%%%%%%%%%%%%%%%%%%%%%%%%%%%%%%%%5
%%%%%%%%%%%%%%%%%%%%%%%%%%%%%%%%%%%%%%%%%5
%%%%%%%%%%%%%%%%%%%%%%%%%%%%%%%%%%%%%%%%%5
%%%%%%%%%%%%%%%%%%%%%%%%%%%%%%%%%%%%%%%%%5
%%%%%%%%%%%%%%%%%%%%%%%%%%%%%%%%%%%%%%%%%5

\section{\DomainContent~Analysis} \label{se:content-analysis}

\subsection{Syntax and Semantics of the Programming Language }

\subsubsection{Loopless Programs}
are generated by the following syntax:
\begin{center}
$
\begin{array}{lll}
e & :: & var.f\mid var\mid\nil\,\,\,\,\,\,\,\,\,\,\,\,\,\,\,\,\,\,\,\,\,\,(f\in\tauField,\, o_{var}\in\tauVar)\\
b & :: &(e_{1}=e_{2})\mid\,\sim b\mid(b_{1}\, and\, b_{2})\mid(b_{1}\, or\, b_{2}) \mid \true \mid \false\\
S & :: & var_{1}:=e_{2}\mid var_{1}.f:=e_{2}\mid skip\mid S_{1};S_{2}
\mid var:=new \mid dispose(var)  \mid \\
 & & if\,\, b\,\, then\,\, S_{1}\,\, fi
  \mid if\,\, b\,\, then\,\, S_{1}\,\, else\,\, S_{2}\,\, fi \mid assume(b)
\end{array}
$
\end{center}
 %We denote by $Variables$ and $Fields$ the corresponding sets of
%names and by
 Let $Exp$
denote the set of expressions $e$ and $Bool$ denote the set of Boolean
expressions $b$.
To define the semantics of pointer and Boolean expressions, we extend  $f^{\mm}$
by $f^{\mm}(\err)=\err$ for every $f\in \tauField$. We define
$\mathcal{E}_e(\mm):Exp\,\,\to \Addr^\mm\cup\{\nil,\err\}$ and $\mathcal{B}_b(\mm):Bool\to\{o_\true,o_\false,\err\}$
(with $\err\not\in M$):
\begin{center}
$
\begin{array}{llll}
\mathcal{E}_{var}(\mm) & =o_{var}^{\mm},  \mbox{ if }o_{var}^{\mm}\in Alloc^{\mm}\ &
\mathcal{B}_{e_{1}=e_{2}}(\mm) & =\err  \mbox{ if }\mathcal{E}_{e_{i}}(\mm)=\err, i\in\{1,2\}\\
\mathcal{E}_{var}(\mm) & =\err,  \mbox{ if }o_{var}^{\mm}\not\in Alloc^{\mm}&
\mathcal{B}_{e_{1}=e_{2}}(\mm) & =o_\true,  \mbox{ if }\mathcal{E}_{e_{1}}(\mm)=\mathcal{E}_{e_{2}}(\mm)\\
\mathcal{E}_{ar.f}(\mm) &=f^{\mm}(\mathcal{E}_{var}(\mm))&
\mathcal{B}_{e_{1}=e_{2}}(\mm) & =o_\false,  \mbox{ if }\mathcal{E}_{e_{1}}(\mm)\not=\mathcal{E}_{e_{2}}(\mm)
\end{array}
$\end{center}
$\mathcal{B}$ extends naturally w.r.t. the Boolean connectives.

The operational semantics of the programming language is:
For any command $S$, if $\mathcal{E}$ or $\mathcal{B}$
give the value $\err$, then $\left\langle S,\mm\right\rangle \leadsto\abo$.
Otherwise, the semantics is as listed below.
First we assume that
in the memory structures involved
all relation symbols either belong to $\tauField$, are ghost symbols
or are the required symbols of memory structures ($Alloc$, $Aux$, etc.).
\begin{enumerate}
\item $\left\langle skip,\mm\right\rangle \leadsto\mm$.
\item %If $\mathcal{E}\lb e_{2}\rb(\mm)]=\err$, then $\left\langle var_{1}:=e_{2},\mm\right\rangle \leadsto\abo$;\\
%otherwise,
$\left\langle var_{1}:=e_{2},\mm\right\rangle \leadsto[\mm\mid o_{var_{1}}^{\mm}\mbox{ is set to }\mathcal{E}_{e_{2}}(\mm)]$.
\item $\left\langle var:=new,
\mm\right\rangle\leadsto$
$[\mm\mid \mbox{For some }t\in\MemPool^\mm,$ \\$ t\mbox{ is moved to }Alloc^{\mm}\mbox{ and }o_{var}^{\mm}\mbox{ is set to }t],$
\item If $o_{var}^{\mm}\not\in Alloc^{\mm}$, $\left\langle dispose(var),\mm\right\rangle \leadsto\abo$;\\
otherwise $\left\langle dispose(var),\mm\right\rangle \leadsto[\mm\mid o_{var}^{\mm}\mbox{ is removed from }Alloc^{\mm}]$.
\item $\left\langle S_{1};S_{2},\mm\right\rangle \leadsto\left\langle S_{2},\left\langle S_{1},\mm\right\rangle \right\rangle$
\item
$\left\langle if\,\, b\,\, then\,\, S_{\true}\,\, else\,\, S_{\false},\mm\right\rangle \leadsto\left\langle S_{tv},\mm\right\rangle $
where $tv=\mathcal{B}_{b}(\mm)$.
\item
$\left\langle if\,\, b\,\, then\,\, S\,\, ,\mm\right\rangle \leadsto
\left\langle if\,\, b\,\, then\,\, S\,\, else\,\, skip\,\, fi ,\mm\right\rangle$.
\item
If $\mathcal{B}_{b}(\mm)=\true$, then
$\left\langle assume(b), \mm\right\rangle \leadsto \mm$;\\
otherwise $\left\langle assume(b), \mm\right\rangle \leadsto \abo$.
\end{enumerate}
If $\mm$ is
a memory structure and $\left\langle S,\mm\right\rangle \leadsto\mm'$,
then $\mm'$ is a memory structure.

Now consider a relation symbol e.g. $ELst$.
If $\left\langle S,\mm\right\rangle \leadsto\mm'$,
then we want to think of $ELst^{\mm}$ and $ELst^{\mm'}$ as the employee list
before and after the execution of $S$. However, the constraints that
$ELst^{\mm}$ and $ELst^{\mm'}$ are lists and that
$ELst^{\mm'}$ is indeed obtained from from $ELst^{\mm}$ by running $S$
will be expressed as formulae.
In the $\leadsto$ relation, we allow any values for $ELst^\mm$ and $ELst^{\mm'}$.

For any tuple $\bar{R}$ of relation symbols which do not belong to $\tauField$,
are not ghost symbols
and are not the required symbols of memory structures ($Alloc$, $Aux$, etc.),
we extend $\leadsto$ as follows: if $\left\langle S,\mm\right\rangle \leadsto\mm'$,
then $\left\langle S,\left\langle \mm,\bar{R}^{\mm}\right\rangle \right\rangle \leadsto\left\langle\mm',\bar{R}^{\mm'}\right\rangle$,
for any tuples $\bar{R}^{\mm}$ and $\bar{R}^{\mm'}$.

\subsubsection{Programs with Loops} are represented as hybrids of the programming language for loopless code
and control flow graphs.
 \begin{definition}[Program]\label{def:program}
 A {\em program} is $G = \left\langle V, E,\linit, shp, \dmnCnt,\lambda\right\rangle$
 such that
 $G=(V,E)$ is a directed graph with no multiple edge but possibly containing self-loops, $\linit \in V$ has in-degree $0$,
  $shp:V\to \SLls$, $\dmnCnt:V\to \L(\tau)$ are functions, and $\lambda$ is a function from $E$ to the set of loopless programs.
 \end{definition}

\begin{tabular}{ll}
Here is the code $S$ from the introduction:   &
\multirow{5}{*}{\ \ \ \ \ \ \ \ \ \ \ \ \ \ \ \ \ \ \ \ \ \ \ \ \ \ \  \includegraphics[scale=0.80]{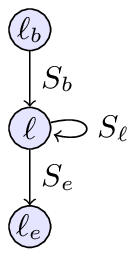}}\\
\end{tabular}\\
\begin{tabular}{ll}
$V=\{\ell_{b},\ell_l,\ell_{e}\}$ & $E=\{(\ell_{b},\ell_l),(\ell_l,\ell_l),(\ell_l,\ell_{e})\}$ \\
$\lambda(\ell_{b},\ell_l)=S_{b}$ & $\linit=\ell_{b}$ \\
 $\lambda(\ell_l,\ell_l)=assume(\sim(e\,=\,\nil));S_{\ell_l}$\\
$\lambda(\ell_l,\ell_e)=assume(e\,=\,\nil);S_e$
\end{tabular}
\\ \ \\
$S_{b}$, $S_{\ell_l}$ and $S_{e}$ denote
 the three loopless code blocks which are respectively the code block before the loop, inside the loop and after the loop.
 The annotations $shp$ and $\dmnCnt$ are described in Section \ref{se:general-meth-example}.

The semantics of programs derive from the semantics of loopless programs and is given in terms of program paths.
Given a program $G$, a {\em path in $G$} is a finite sequence of directed edges
$e_1,\ldots,e_t$ such that for all $1\leq i\leq t-1$, the tail of $e_i$
is the head of $e_{i+1}$.
A path may contain cycles.
% We extend the $\leadsto$ relation to paths in programs:
 \begin{definition}[$\leadsto^*$ for paths]
 Given a program $G$, a path $P$ in $G$,
 and memory structures $\mm_1$ and $\mm_2$
 we define whether
 $\left\langle P, \mm_1\right\rangle\leadsto^* \mm_2$
 holds inductively.
\ifiFM 
\else
 \begin{itemize}
 \item 
\fi
If $P$ is empty, then $\left\langle P, \mm_1\right\rangle\leadsto^* \mm_2$ iff $\mm_1=\mm_2$.
\ifiFM
\else
 \item 
\fi
If $e_t$ is the last edge of $P$, then $\left\langle P, \mm_1\right\rangle\leadsto^* \mm_2$ iff
 there is $\mm_3$ such that
 $\left\langle P\backslash\{e_t\}, \mm_1 \right\rangle\leadsto^* \mm_3$
 and $\left\langle \lambda(e_t), \mm_1 \right\rangle\leadsto^* \mm_3$.
  $P\backslash\{e_t\}$ denotes the path obtained from $P$ by removing the last edge $e_t$.
\ifiFM
\else
 \end{itemize}
\fi
 \end{definition}

\subsection{Hoare-style Proof System}

Now we are ready to state our two-step verification methodology that we formulated in Section \ref{se:intro} precisely.
Our methodology assumes a program $P$ as in Definition~\ref{def:program} as input (ignoring the $shp$ and $\dmnCnt$ functions for the moment).

\textbf{I. Shape Analysis.}
The user annotates the program locations with SL formulae from $\SLls$ (stored in the $shp$ function of $P$).
Then the user proves the validity of the $\SLls$ annotations, for example, by using techniques from~\cite{pr:BCO05}.

\textbf{II. Content Analysis.}
The user annotates the program locations with $\L$-formulae that she wants to verify (stored in the $\dmnCnt$ function of $P$).
We point out that an annotation $\dmnCnt(\ell)$ can use the concepts occurring in $\alpha(shp(\ell))$ (recall that $\alpha: \SLls \rightarrow \L$ maps SL formulae to $\L$-formulae).

In the rest of the paper we discuss how to verify the $\dmnCnt$ annotations.
In Section~\ref{se:methodology} we describe how to derive a verification condition for every program edge.
The verification conditions rely on the backwards propagation function $\Theta$ for $\L$-formulae which we introduce in Section~\ref{se:backwards-example}.
The key point of our methodology is that the validity of the verification conditions can be discharged automatically by a satisfiability solver for $\ctwotree$-formulae.
We show that all the verification conditions are valid if and only if $\dmnCnt$ is inductive.
Intuitively, $\dmnCnt$ being inductive ensures that the annotations $\dmnCnt$ can be used in an inductive proof to show that all reachable memory structures indeed satisfy the annotations $\dmnCnt(\ell)$ at every program location $\ell$ (see Definition~\ref{def:inductive-annotation} below).

%
%For concreteness, we have defined an automatic translation from $\SLls$ annotations to $\L$-formulae in Section \ref{se:expressing-shape-invariants}.
%However, we want to point out that our methodology is independent of the concrete choice of the shape predicates.
%We only require a translation of shape invariants to $\L$-formulae, as discussed in Section \ref{se:expressing-shape-invariants}.
%
%The user translates the shape invariants into $\L$-formulae
%
%For concreteness, we have stated the SL $\SLls$ as a possible language for shape annotations in Section~\ref{se:SLls}.

%Note that $\Psi$ is a that is \emph{precise} for loopless code.

\subsection{Content Verification}\label{se:methodology}

 We want to prove that, for every initial memory structure $\mm_1$
 from which the computation satisfies $shp$ and which satisfies the \domainContent~pre-condition $\dmnCnt(\linit)$,
 the computation satisfies $\dmnCnt$. Here are the corresponding verification conditions, which annotate the vertices of $G$:

\begin{definition}[Verification conditions]
 Given a program $G$, $VC$ is the function from $E$ to $\L$
 given for $e=(\ell_0,\ell)$ by
\ifiFM
 $
    VC(e) = \neg \big[ \beta(shp(\ell_0))\land tr(\dmnCnt(\ell_0))  \land tr\left( \Theta_{\lambda(e)}\big(\alpha(shp(\ell)) \land \neg  \dmnCnt(\ell)\big)\right)\big]
$
\else
\begin{eqnarray*}
    VC(e) &=& \neg \big[ \beta(shp(\ell_0))\land tr(\dmnCnt(\ell_0))  \land tr\left( \Theta_{\lambda(e)}\big(\alpha(shp(\ell)) \land \neg  \dmnCnt(\ell)\big)\right)\big]
 \end{eqnarray*}
\fi
 $VC(e)$ {\em holds} if $VC(e)$ is a tautology over memory structures ($\top \models_m VC(e)$).
 %for every extended memory structure $\nn$, tuple of relations $\bar{U}=(U_R : R\in \tau\backslash \tauField)$ and tuple $\bar{d}$ we have
%\[\left\langle \nn, \bar{U}, \bar{d} \right\rangle \models VC(e)\]
% $\bar{U}$ interprets the primed relation symbols in $\tau_{ext,Y}'$.
 \end{definition}
 $\Theta$ is discussed in Section \ref{se:backwards--}. As we will see, $VC(\ell_0,\ell)$ expresses that
 when running the loopless program $\lambda(e)$ when the memory satisfies the the annotations of $\ell_0$,
 and when the shape annotation of $\ell$ is at least partly true (i.e., when $\alpha(shp(\ell))$),
 the content annotation of $\ell$ holds.

 Let $J$ be a set of  memory structures.
 For a formula in $\ctwotree$ or $\L$, we write $J\models\varphi$
 if, for every $\mm\in J$, $\mm\models \varphi$.
 Let $Init$ be a set of memory structures.
 \begin{definition}[Inductive program annotation]
 \label{def:inductive-annotation}
 Let $f:V\to \ctwotree$.
 We say $f$ is {\em inductive} for $Init$ if (i) $Init\models f(\ell_{init})$, and (ii)
 for every edge $e=(\ell_1,\ell_2)\in E$ and memory structures
 $\mm_1$ and $\mm_2$ such that $\mm_1\models f(\ell_1)$ and $\left\langle \lambda(e),\mm_1\right\rangle \leadsto \mm_2$,
 we have $\mm_2\models f(\ell_2)$.
 We say $shp$ is inductive for $Init$ if the composition $shp\circ\beta:V\to\ctwotree$ is inductive for $Init$.
 We say {\em $\dmnCnt$ is inductive for $Init$ relative to $shp$}
 if $shp$ is inductive for $Init$ and
 $g:V\to\ctwotree$ is inductive for $Init$,
 where $g(\ell)=tr(\dmnCnt(\ell))\land \beta(shp(\ell))$.
 \end{definition}

 \vspace{4pt}
 \begin{theorem}[Soundness and Completeness of the Verification Conditions] \label{th:soundness-completeness}
 Let $G$ be a program such that $shp$ is inductive for $Init$ and $Init\models cnt(\ell_{init})$.
The following statements are equivalent:\\
 (i)  For all $e\in E$, $VC(e)$ holds.\ \ \ \ \ (ii)  $\dmnCnt$ is inductive for $Init$ relative to $shp$.
 \end{theorem}
\ifiFM \else
 We make the notion of a computation satisfying the verification conditions precise using the following definition:
\fi
 \begin{definition}[$Reach(\ell)$]
 Given a program $G$, a node $\ell\in V$, and a set $\Init$ of memory structures, % such that $(\nn_0)_{cur}=(\nn_{init})_{\gho}$ for each $\nn_0\in \Init$,
 $Reach(\ell)$ is the set of  memory structures $\mm$ for which there is
 $\mm_{init} \in \Init$ and a path $P$ in $G$ starting at $\linit$
 such that $\left\langle P,\mm_{init}\right\rangle\leadsto^* \mm$.
 \end{definition}
 In particular, $Reach(\ell_{init})=Init$.
 The proof of Theorem \ref{th:soundness-completeness} and its consequence Theorem \ref{th:soundness-2} below
 are given in 
\ifiFM
the full version \cite{full-version}.
\else in Appendix \ref{app:soundness-completeness}.
\fi
 \begin{theorem}[Soundness of the Verification Methodology] \label{th:soundness-2}
 Let $G$ be a program such that $shp$ is inductive for $Init$ and $Init\models cnt(\ell_{init})$.
 If for all $e\in E$,
 $VC(e)$ holds, then for $\ell\in V$, $Reach(\ell)\models \dmnCnt(\ell)$.
 \end{theorem}
\subsection{Running Example: General Methodology}  \label{se:general-meth-example}
To verify the correctness of the code $S$, the $shp$ and $\dmnCnt$ annotations must be provided.
The shape annotations of program $S$ are:
\ifiFM
$shp(\ell_b)=\ls(eHd,\nil)*\ls(pHd,\nil)$, $shp(\ell_e)=(proj=pHd) \bv \ls(eHd,\nil)*\ls(pHd,\nil)$, and $shp(\ell_l)= \varphi_{\ell_l}$,
where 
$\varphi_{\ell_l}$ is from Section \ref{se:shape-inv-example}.
\else
\[
\begin{array}{lll}
shp(\ell_b)&=&\ls(eHd,\nil)*\ls(pHd,\nil)\\
shp(\ell_l)&=& \varphi_{\ell_l}\\
shp(\ell_e)&=&(proj=pHd) \bv
 \ls(eHd,\nil)*\ls(pHd,\nil)
\end{array}
\]
$\varphi_{\ell_l}=\ls(eHd,\nil)*\ls(e,\nil)*\ls(pHd,e)$ was considered in Section \ref{se:shape-inv-example}.
\fi

The three \domainContent~annotations require that the system invariants
$\varphi_{invariants}$ from Section \ref{se:expressing-content-invariants}
hold. The post-condition additionally requires that $\varphi_{p-assgn}$ and $\varphi_{lists-updts}$ hold. Recall $\varphi_{p-assgn}$
states that every employee which was not assigned a project, is assigned to $o_{proj}$.
$\varphi_{lists-updts}$ states that the content of the two lists remain unchanged, except that the project $o_{proj}$
is inserted to $PLst$.

In order to interact with the translations $\alpha(shp(\cdots))$ of the shape annotations,
we  need to related the $P_i$ to the concepts $ELst$ and $PLst$.
In Section  \ref{se:shape-inv-example} we defined $\psi_l$, which relates the $P_i$ generated by $\alpha$
on $shp(\ell_l)$.
\ifiFM
We have 
$\psi_{\ell_b}= \psi_{\ell_e}= P_1 \equiv ELst \land P_2 \equiv PLst$.
Then 
$\dmnCnt(\ell_b)=\psi_{\ell_b}\land\varphi_{invariants}$, 
$\dmnCnt(\ell_l)=\psi_{\ell_l}\land\varphi_{invariants}\land \varphi_{lists-updt}\land \varphi_{p-as-\ell_l}$, and 
$\dmnCnt(\ell_e)=\psi_{\ell_e}\land\varphi_{invariants}\land \varphi_{lists-updt}\land \varphi_{p-assgn}$, where
$\varphi_{p-as-\ell_l} = P_1\sqcap\exists wrkFor_{\gho}.o_{\nil} \equiv P_1\sqcap\exists wrkFor.o_{proj}$. 
\else
\begin{eqnarray*}
\psi_{\ell_b}= \psi_{\ell_e}&=&P_1 \equiv ELst \land P_2 \equiv PLst\\
\dmnCnt(\ell_b)&=&\psi_{\ell_b}\land\varphi_{invariants}\\
\dmnCnt(\ell_l)&=&\psi_{\ell_l}\land\varphi_{invariants}\land \varphi_{lists-updt}\land \varphi_{p-as-\ell_l}\\
\dmnCnt(\ell_e)&=&\psi_{\ell_e}\land\varphi_{invariants}\land \varphi_{lists-updt}\land \varphi_{p-assgn}\\
 \varphi_{p-as-\ell_l} &=& P_1\sqcap\exists wrkFor_{\gho}.o_{\nil} \equiv
 P_1\sqcap\exists wrkFor.o_{proj}
 \end{eqnarray*}
\fi
$\varphi_{p-as-\ell_l}$ states that, in the part of $ELst$ containing the employees visited so far in the loop,
any employee which was not assigned to a project at the start of the program (i.e., in the ghost version of $wrkFor$)
is assigned to the project $proj$.
$\varphi_{p-as-\ell_l}$ makes no demands on elements of $ELst$ which have not been reach in the loop so far.
The verification conditions of $G$ are, for each $(\ell_1,\ell_2)\in E$,
\ifiFM
$ VC(\ell_1,\ell_2) = \neg \big[\beta(shp(\ell_1))\land tr(\dmnCnt(\ell_1)) \land
    tr\big(\Theta_{\lambda(l_1,l_2)}(\alpha(shp(\ell_2))\land \neg  \dmnCnt(\ell_2))\big)\big]$.
\else
\[
\begin{array}{lll}
    VC(\ell_1,\ell_2)& =& \neg \big[\beta(shp(\ell_1))\land tr(\dmnCnt(\ell_1)) \land
    tr\big(\Theta_{\lambda(l_1,l_2)}(\alpha(shp(\ell_2))\land \neg  \dmnCnt(\ell_2))\big)\big]\\
\end{array}
\]
\fi
The verification conditions $VC(e)$ express
that the loopless programs on the edges $e$ of $G$ satisfy their annotations.
To prove the correctness of $G$ w.r.t. $VC(e)$ using Theorem \ref{th:soundness-2},
we prove that $VC(e)$, $e\in E$, hold, in order 
\ifiFM
to get as a conclusion that
$Reach(\ell)\models \dmnCnt(\ell)$, for all $\ell\in V$.
\else
to get:
\begin{conclusion}
$Reach(\ell)\models \dmnCnt(\ell)$, for all $\ell\in V$.
\end{conclusion}
\fi

\subsection{Backwards Propagation and the Running Example}\label{se:backwards--}\label{se:backwards-example}
Here we shortly discuss the backwards propagation of a formula along a loopless program $S$.
Let $\left\langle S, \mm_1\right\rangle\leadsto \mm_2$ where $\mm_1$ and $\mm_2$
are memory structures over the same vocabulary $\tau$.
E.g., in our running example, for $i=1,2$, $\mm_i$ is
\ifiFM
$\big\langle M,$ $ ELst^{\mm_i}, next^{\mm_i}, mngBy^{\mm_i}, \cdots,
 ELst_{\gho}^{\mm_i}, next_{\gho}^{\mm_i},  \cdots, Alloc^{\mm_i}, Aux^{\mm_i} \cdots  \big\rangle$. \linebreak
\else
\[\begin{array}{l}
\big\langle M, ELst^{\mm_i}, next^{\mm_i}, mngBy^{\mm_i}, \cdots,\\
 ELst_{\gho}^{\mm_i}, next_{\gho}^{\mm_i},  \cdots, Alloc^{\mm_i}, Aux^{\mm_i} \cdots  \big\rangle
  \end{array}
\]
\fi
We will show how to translate a formula for $\mm_2$ to a formula for an extended $\mm_1$.
Fields and variables in $\mm_2$ will be translated by the backwards propagation
 into expressions
involving elements of $\mm_1$. For ghost symbols $s_{\gho}$, $s_{\gho}^{\mm_1}$ will be used instead of $s_{\gho}^{\mm_2}$
since they do not change during the run of the program. Let $\tau^{rem}\subseteq \tau$ be the set of the remaining symbols,
i.e.  the symbols of $\tau \setminus (\{PossibleTargets,\MemPool \}\cup \tauField)$ which are not ghost symbols,
for example $ELst$, but not $ELst_{\gho}$, $next$ or $mngBy$.
We need the result of the backwards propagation to refer to the interpretations of symbols in $\tau^{rem}$
from $\mm_2$ rather than $\mm_1$.
Therefore,
these interpretations are copied as they are from $\mm_2$ and added to $\mm_1$ as follows.
For every $R\in \tau^{rem}$, we add a symbol $R^{ext}$ for the copied relation.
We denote by $(\bar{R}^{ext})^{\mm_1}$ the tuple
$
 \big((R^{ext})^{\mm_1} : (R^{ext})^{\mm_1}=R^{\mm_2}\mbox{ and }R\in\tau^{rem}\big)
$
Let $\tau^{ext}$ extend $\tau$ with $R^{ext}$ for each $R\in\tau^{rem}$.
The backwards propagation updates the fields and variables according to the loopless code.
Afterwards, we substitute the symbols $R\in\tau^{rem}$ in $\varphi$ with the corresponding $R^{ext}$.
We present here a somewhat simplified version of the backwards propagation lemma. The precise version is similar in spirit
and is 
\ifiFM
in the full version \cite{full-version}
\else
given in Appendix \ref{app:backward}.
\fi
\begin{lemma}[Simplified]\label{lem:backwards--} \label{se:backwards-ghost--}
Let $S$ be a loopless program, let
$\mm_1$ and $\mm_2$ be  memory structures, and
 $\varphi$ be an $\L$-formula over $\tau$.
\ifiFM
(1)
\else
\begin{enumerate}
 \item 
\fi
If $\left\langle S, \mm_1\right\rangle \leadsto \mm_2$,
 then:
 $\mm_2\models \varphi$ iff
$\left\langle \mm_1, (\bar{R}^{ext})^{\mm_1}\right\rangle \models \Theta_{S}(\varphi)$.
\ifiFM
(2) 
\else
\item 
\fi 
If $\left\langle S, \mm_1\right\rangle\leadsto \abo$, then
$\left\langle \mm_1, (\bar{R}^{ext})^{\mm_1},\right\rangle \not\models \Theta_{S}(\varphi)$.
\ifiFM \else
\end{enumerate}
\fi
\end{lemma}

As an example of the backwards propagation process,
we consider a formula from Section \ref{se:general-meth-example},
which is part of the \domainContent~annotation of $\ell_l$
and perform the backwards propagation on the loopless program inside the loop:
\ifiFM
$\varphi_{p-as-\ell_l} =
P_1 \sqcap\exists wrkFor_{\gho}.o_{\nil} \equiv
P_1 \sqcap\exists wrkFor.o_{proj}$
\else
\[
\begin{array}{lll}
\varphi_{p-as-\ell_l} &=&
P_1 \sqcap\exists wrkFor_{\gho}.o_{\nil} \equiv
P_1 \sqcap\exists wrkFor.o_{proj}
\end{array}
\]
\fi
Since $next$ does not occur in $\varphi_{p-as-\ell_l}$, backwards propagation of $\varphi_{p-as-\ell_l}$
over $e := e.next$ does not change the formula (however $\alpha(shp(\ell_l))$
by this command). The backwards propagation of the $if$ command gives
\ifiFM
$\Psi_{S_{\ell_l}}(\varphi_{p-as-\ell_l}) = 
 \big(\neg(\exists wrkFor^-.o_e \equiv o_\nil)  \land
 \varphi_{p-as-\ell_l} \big) \lor 
 \exists wrkFor^-.o_e \equiv o_\nil\hphantom{)}\land
 \Psi_{e.wrkFor := proj}(\varphi_{p-as-\ell_l})\big)$ and
$\Psi_{e.wrkFor := proj}(\varphi_{p-as-\ell_l}) $  \linebreak
$  =
 P_1 \sqcap\exists wrkFor_\gho.o_{\nil} \equiv  
  P_1\sqcap\exists {((wrkFor\backslash (o_{e} \times \top)) \cup (o_e,o_{proj}))}.o_{proj}$. \linebreak
\else
\[
\begin{array}{ll}
 \Psi_{S_{\ell_l}}(\varphi_{p-as-\ell_l}) = &
 \big(\neg(\exists wrkFor^-.o_e \equiv o_\nil)  \land
 \varphi_{p-as-\ell_l} \big) \lor \\
 &\exists wrkFor^-.o_e \equiv o_\nil\hphantom{)}\land
 \Psi_{e.wrkFor := proj}(\varphi_{p-as-\ell_l})\big)\\
\Psi_{e.wrkFor := proj}(\varphi_{p-as-\ell_l}) =&
 P_1 \sqcap\exists wrkFor_\gho.o_{\nil} \equiv  \\
 & P_1\sqcap\exists {((wrkFor\backslash (o_{e} \times \top)) \cup (o_e,o_{proj}))}.o_{proj}
\end{array}
\]
\fi
$\Psi_{e.wrkFor := proj}(\varphi_{p-as-\ell_l})$ is obtained from $\varphi_{p-as-\ell_l}$ by substituting the $wrkFor$
role with the correction  $((wrkFor\backslash (o_{e} \times \top)) \cup (o_e,o_{proj}))$
which updates the value of $o_e$ in $wrkFor$ to $proj$.
$\Phi_{S_{\ell_l}}(\varphi_{p-as-\ell_l})$ is obtained from $\Psi_{S_{\ell_l}}(\varphi_{p-as-\ell_l})$
by subtituting ${P_1}$ with ${P_1}^{ext}$.
$\Theta$ is differs from $\Phi$ from technical reasons related to
aborting computations 
\ifiFM
(see the full version \cite{full-version}).
\else
(see Appendix \ref{app:backward}).
\fi

\input{related}

\bibliographystyle{plain}

\newpage

 \appendix
%
% \setcounter{tocdepth}{4}
% \tableofcontents
%
\ifiFM
\else
 \input{appendix-sl}

\input{appendix-nexptime}

\input{appendix-translations}
 \input{appendix-cyclic}
 \input{appendix-backwards}

\input{appendix-soundness}
\fi

\end{document}

%% file: introduction.tex
\newcommand{\nop}[1]{{}}

\section{Introduction}\label{se:intro}

The manipulation and storage of complex information in imperative programming languages is often achieved by dynamic data structures.
The verification of programs with dynamic data structures, however, is notoriously difficult, and is a highly active area of current research.
While much progress has been made recently in analyzing and verifying the \emph{shape} of dynamic data structures, most notably by separation logic (SL)~\cite{pr:Reynolds02,ar:IO01}, the \emph{content} of dynamic data structures has not received the same attention.

In contrast, disciplines as data\-bases, modeling and knowledge representation have developed highly-successful theories for \emph{content representation and verification}.
These research communities typically model reality by classes and binary relationships between these classes.
For example, the database community uses \emph{entity-relationship (ER)} diagrams, and \emph{UML} diagrams have been studied in requirements engineering.
Content representation in the form of UML and ER has become a central pillar of \emph{industrial software engineering}.
In complex software projects, the source code is usually accompanied by \emph{design documents} which provide extensive documentation and models of data structure content.
This documentation is both an opportunity and a challenge for program verification. Recent hardware verification papers have demonstrated how design diagrams can be integrated into an industrial verification workflow \cite{pr:JLTT08}.

In this paper, we propose the use of {\em Description Logics} (DLs) for the formulation of content specifications.
DLs are a well established and highly popular family of logics for representing knowledge in artificial intelligence \cite{bk:Baader2003}.
In particular, DLs allow to precisely model and reason about UML and ER diagrams~\cite{Berardi200570,er2007}.
DLs are mature and well understood, they have good algorithmic properties and
have efficient reasoners.
DLs are very readable and form a natural base for developing specification languages. For example, they are the logical backbone of the Web Ontology Language (OWL) for the Semantic Web~\cite{owl2-overview}.
DLs vary in expressivity and complexity, and are usually selected according to the expressivity needed to formalize the given target domain.

Unfortunately, the existing content representation technology cannot be applied directly for the verification of content specifications of pointer-manipulating programs.
This is to due the strict separation between high-level content descriptions such as UML/ER and the way data is actually stored.
For example, query languages such as SQL and Datalog provide a convenient abstraction layer for formulating data queries while ignoring how the database is stored on the disk.
In contrast, programs with dynamic data structures manipulate their data structures directly.
Moreover, database schemes are usually static while a program may change the content of its data structures over time.

The main goal of this paper is to develop a verification methodology that allows to employ DLs for formulating and verifying content specifications of pointer-manipulating programs.
We propose a two-step Hoare-style verification methodology:
First, existing shape-analysis techniques are used to derive shape invariants.
Second, the user strengthens the derived shape invariants with content annotations; the resulting verification conditions are then checked automatically.
Technically, we employ a very expressive DL (henceforth called $\L$), based on the so called $\mathcal{ALCHOIF}$, which we specifically tailor to better support reasoning about complex pointer structures.
For shape analysis we rely on the SL fragment from \cite{pr:BCO05}.
In order to reason automatically about the verification conditions involving DL as well as SL formulae, we identify a powerful decidable logic $\ctwotree$ which incorporates both logics~\cite{pr:CharatonikWitkowski13}.
We believe that our main contribution is conceptual, integrating these different formalisms for the first time.
While the current approach is semi-manual, our long term goal is to increase the automatization of the method.

\newpage 
\textbf{Overview and Contributions:}
\begin{itemize}
\item  In Section \ref{se:logics}, we introduce our formalism.
    In particular, we formally define \emph{memory structures} for representing
    the heap and we study the DL $\L$ as a formalism for expressing \emph{content properties} of memory structures.

\item In Section \ref{se:logics}, we further present the building blocks for our verification methodology:
    We give an embedding of $\L$ and an embedding of a fragment of the SL from \cite{pr:BCO05} into $\ctwotree$ (Lemmata~\ref{lem:dl-to-ct2} and~\ref{lem:sl-to-dl-2-alpha}).
    Moreover, we give a complexity-preserving reduction of satisfiability of $\ctwotree$ over memory structures to finite satisfiability of $\ctwotree$ (Lemma \ref{lem:fin-sat-mem}).

\item In Section \ref{se:content-analysis}, we describe a program model for sequential imperative heap-manipulating programs without procedures.
    Our main contribution is a Hoare-style proof system for verifying content properties on top of (already verified) shape properties stated in SL.

\item Our main technical result is a precise backward-translation of content properties along loop-less code (Lemma~\ref{lem:backwards--}).
    This backward-translation allows us to reduce the inductiveness
    of the Hoare-annotations to satisfiability in $\ctwotree$.
    Theorem~\ref{th:soundness-completeness} states the soundness and completeness of this reduction.
\end{itemize}

\subsection{Running Example: Information System of a Company} \label{se:running-intro}
\begin{wrapfigure}[7]{r}{80mm}
\begin{center}
\vspace{-42pt}
\includegraphics[scale=0.9]{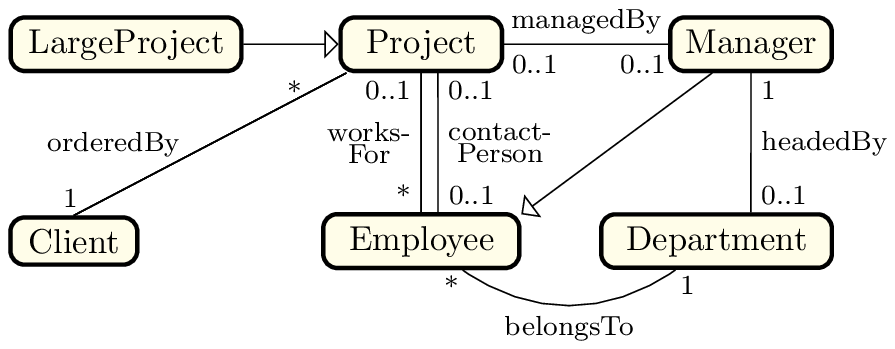}
\end{center}
\end{wrapfigure}\vspace{3pt}
Our running example will be a simple information system for a company with the following UML diagram:.
The UML gives the relationships between entities in the information system, but says nothing regarding the
implementations of the data structures that hold the data.
We focus mostly on projects, and on the employees and managers which work on them.
Here is an informal
description of the programmers' intention.
The employees and projects are stored in two lists, both using the $next$ pointer.
The heads of the two lists are $pHd$ and $eHd$ respectively.
Here are some properties of our information system.
(i)-(iii) extends the UML somewhat. (iv)-(vi) do not appear in the UML,
but can be expressed in DL:

\begin{enumerate}[(i)]
 \item Each employee in the list of employees has a pointer $wrkFor$ to a project on
the list of projects, indicating the project that the employee is working on (or
to $\nil$, in case no project is assigned to that employee).
\item Each project in the list has a pointer $mngBy$ to the employee list,
indicating the manager of the project (or to $\nil$, if the project doesn't have one).
\item Employees have a Boolean field $isMngr$ marking them as managers, and only they can manage projects.
 \item The manager of a project works for the project.
 \item At least 10 employees work on each large project. 
 \item The contact person for a large-scale project is a manager.
\end{enumerate}
We will refer to these properties as the \emph{system invariants}.

The programmer has written a program $S$ (stated below) for verification.
The programmer has the following intuition about her program:
The code $S$ adds a new project $proj$ to the project list, and assigns to it all employees in the employee list which are not assigned to any project.

\begin{wrapfigure}[9]{r}{50mm}
\begin{center}
\vspace{-52pt}
\begin{minipage}[t]{0.45\textwidth}
\begin{lstlisting}[basicstyle={\small},commentstyle={\color{mygreen}},escapechar={!},mathescape=true,morekeywords={if,while,do,od,fi,else,then,func,end}]
$\ell_{\mathit{b}}:$ proj$\,:=\,$new;
   proj.next$\,:=\,$pHd;
   pHd$\,:=\,$proj;
   e$\,:=\,$eHd;
$\ell_l:$  while $\sim$(e$\,=\,\nil$) do
     if (e.wrkFor$\,=\,\nil$)
     then e.wrkFor$\,:=\,proj$;
     e$\,:=\,$e.next;
   od
$\ell_{\mathit{e}}:$ end;
\end{lstlisting}
\end{minipage}
\vspace{10pt}
\end{center}
\end{wrapfigure}

The programmer wants to verify that the system invariants are true after the execution of $S$,
if they were true in the beginning (1).
Note that \emph{during} the execution of the code, they might not be true!
Additionally the programmer wants to verify that after executing $S$, the project list
has been extended by $proj$, the employee list still contains the same employees and indeed
all employees who did not work for a project before now work for project $proj$ (2).
We will formally prove the correctness of $S$ following our verification
methodology discussed in the introduction.
In Section~\ref{se:expressing-content-invariants} we describe how our DL can be used for specifying the verification goals (1) and (2).
In Section~\ref{se:general-meth-example} we state verification conditions that allow to conclude
the correctness of (1) and (2) for $S$.

%% file: related.tex
\section{Related Work}

\noindent
\textbf{Shape Analysis }
 attracted considerable attention in the literature.
The classical introductory paper to SL \cite{pr:Reynolds02}
presents an expressive SL which turned out to be undecidable.
We have restricted our attention to the
better behaved fragment in
\cite{pr:BCO05}.
The work on SL focuses mostly on shape rather than content in our sense.
SL has been extended to object oriented languages, cf. e.g. 
\cite{ar:ParkinsonBierman2008,ar:ChinDavidNguyenHuuQin2008},
where shape properties similar to those studied in the non objected oriented case are the focus, and
the main goal is to overcome difficulties introduced by the additional features of OO languages.
Other shape analyses could be potential candidates for integration in our methodology.
\cite{YorshRepsSagiv04} use 3-valued logic to perform shape analysis.
Regional logic is used to check correctness of program with shared dynamica memory areas \cite{ar:BNR13}.
%Role logic is a notation for describing properties of relational structures
%in shape analysis, databases, and knowledge bases.  \cite{arXiv:KuncakRinard04}.
\cite{pr:HLRSV13} uses nested tree automata to represent the heap.
\cite{pr:MadhusudanXiaokang11} combines monadic second order logic with SMT solvers.

\noindent \textbf{Description Logics }
 have not been considerd for verification of programs with
dynamically allocated memory, with the exception of \cite{pr:GM05}
whose use (mostly undecidable) DLs to express shape-type invariants,
ignoring content information.  In~\cite{pr:CalvaneseOrtizSimkus13} the
authors consider verification of loopless code (transactions) in graph
databases with integrity constraints expressed in DLs. Verification of
temporal properties of dynamic systems in the presence of DL knowledge
bases has received significant attention~(see
\cite{DBLP:conf/frocos/BaaderZ13,DBLP:conf/kr/GiacomoLP12}
and their references). \emph{Temporal Description Logics}, which
combine classic DLs with classic temporal logics, have also received
significant attention in the last decade (see
\cite{DBLP:conf/time/LutzWZ08} for a survey).

\noindent
\textbf{Related Ideas. } Some recent papers have studied verification strategies which
use information beyond the semantics of the source code. E.g., 
\cite{pr:JLTT08} is using diagrams from design documentation to support verification. \cite{pr:DVZ13,pr:ABFRR13} infer the intended use of program variables
to guide a program analysis.
Instead of starting from code and verifying its correctness, \cite{pr:HAFRS10} explores how to declaratively specify data structures with sharing and how to automatically generate code from this specification.
Given the importance of both
DL as a formalism of content representation and of program verification, and given
that both are widely studied, we were surprised to find little related work. 
However, we believe this stems from large differences between the research in the two communities,
and from the interdisciplinary nature of the work involved.

%%% Local Variables:
%%% mode: latex
%%% TeX-master: "../paper"
%%% End:

%% file: appendix-sl.tex
\section{Separation Logic\label{se:SL-app}}

Here we expand on the treatment of separation logic in the paper.
We have defined the semantics of SL using our memory structures. 
The memory model used in \cite{pr:BCO05} is very similar to our memory structures. We give the standard 
semantics of SL here in terms of heaps and stacks, and relate it to memory structures. It is convenient 
to define first $\SL$, which does not allow list segments, and then extend to $\SLls$. 

The allocated cells of the memory all have the same finite collection of fields (denoted $Fields$).
Let $Add$ and $Val$ be disjoint sets. $Add$ is the set of addresses ({\em locations} in the terminology of
\cite{pr:BCO05}, not to be confused with our use of ``locations in a program'').
$Val$ is a set of values which include $nil$.
The description of the memory consists of two parts, a heap and a stack.
A heap is a partial function $h:Add\arrowfin(Fields\to Val\cup Add)$
which is only defined on a finite subset of $Add$.
A stack is a function from a finite set $Var$
of local variables $s:Var\to Val\cup Add$.

The syntax of $\SL$ is as follows:\\
\[
\begin{array}{l}
var_{i}\in Var,\, i\in\mathbb{N}\\
f_{i}\in Fields,\, i\in\mathbb{N}\\
E::=\nil\mid var_{i}\\
\end{array}\,\,\,\begin{array}{l}
\Pi::=\true\mid E=E\mid E\not=E\mid\Pi\land\Pi\\
\Sigma::=\emp\mid  \Sigma*\Sigma \mid \\
var\mapsto[f_{1}:E_{1},\ldots,f_{k}:E_{k}] \\
\SL\mbox{-formula}::=\Pi\bv\Sigma
\end{array}
\]
When $\Pi=\true$ we write $\Pi\bv\Sigma$ simply as $\Sigma$.
The semantics of $\SL$ is given by a relation $s,h\models\phi$ where
$s\in Stacks,h\in Heaps$. We define $\lb var_{i}\rb s\stackrel{def}{=}s(x)$
and $\lb nil\rb s\stackrel{def}{=}nil$, and:\\
\[
\begin{array}{lcl}
s,h\models true &  & \mbox{always}\\
s,h\models E_{1}=E_{2} & \iffdef & \lb E_{1}\rb s=\lb E_{2}\rb s\\
s,h\models E_{1}\not=E_{2} & \iffdef & \lb E_{1}\rb s\not=\lb E_{2}\rb s\\
s,h\models\Pi_{0}\land\Pi_{1} & \iffdef & s,h\models\Pi_{0}\mbox{ and }s,h\models\Pi_{1}\\
s,h\models var_{1}\mapsto[f_{i}:E_{i}] & \iffdef & h=\lb var_{1}\rb s\to r
 \mbox{ where }r(f_{i})=\lb E_{i}\rb s\\
s,h\models\emp & \iffdef & h=\emptyset\\
s,h\models\Sigma_{1}*\Sigma_{2} & \iffdef & \exists h_{0}h_{1}.\, h=h_{0}*h_{1}\mbox{ and }
s,h_{0}\models\Sigma_{0}\mbox{ and }s,h_{1}\models\Sigma_{1}\\
s,h\models\Pi\bv\Sigma & \iffdef & s,h\models\Pi\mbox{ and }s,h\models\Sigma
\end{array}
\]
where $h_0*h_1$ enforces there is no address in $Add$ on which both $h_{0}$ and $h_{1}$ are defined, and $h=h_{0}*h_{1}$ denotes that $h$ is the union
of $h_{0}$ and $h_{1}$. Additionally, not all fields in $Fields$ need to occur in $var_1\mapsto[f_{i}:E_{i}]$,
and those that do not are
assigned $nil$ implicitly.

$\SLls$ extends $\SL$ by adding list
segments: the syntax is extended by
\[\Sigma::=\cdots\mid\ls(E_{1},E_{2})\]
and the semantics of $\ls(E_{1},E_{2})$ is the least fixed point
of the predicate given by 
\[(E_{1}=E_{2}\land\emp)\lor(E_{1}\not=E_{2}\land\exists y.E_{1}\mapsto[n:y]*\ls(y,E_{2}))\]

\paragraph*{Heaps and Stacks vs. Memory Structures}
The memory model of SL and our memory model are easily
translatable. The distinction between values in $Val$ and addresses in $Add$ does not play a major role
in \cite{pr:BCO05}, so we simplify by setting $Val = \{nil, true,false\}$.

Given $\mm$ we define $s_\mm$ and $h_\mm$ as follows.
$Fields$ is equal to $\tauField$ and $Add = Addresses$
and we have $M=Val \cup Add$
with $nil=o_\nil^\mm$, $true=o_\true^\mm$ and $false=o_\false^\mm$.
For every variable $var_i$, $s(var_i)=o_{var_i}^\mm$.
For every $add\in Add$ we define $h(add)=h_{add}$
as follows: for every $f\in Fields$, $h_{add}(f)=f^\mm(add)$.

Given $s$ and $h$, we define $\mm_{s,h}=\mm$ as follows.
The universe $M_{s,h}$ is $Add\cup \{o_\nil^\mm,o_\true^\mm,o_\false^\mm\}$,
with $o_\nil^\mm=nil$, $o_\true^\mm=true$ and $o_\false^\mm=false$
and $Addresses^\mm = Add$.
The set of addresses on which $h$ is defined is $Alloc^\mm$.
For every field $f\in Fields$,
$f^\mm = h_f$, where $h_f(add)=h(add)(f)$ is the value the field $f$
of address $add$ receives under $h$. For every variable $var_i$,
$o_{var_i}^{\mm_{s,h}} = s(var_i)$.

% The above shows how to transform the memory models into each other
% in a natural way. More precisely,
% if $(s_1,h_1)$  and $(s_2,h_2)$ are isomorphic pairs of stack and heap,
% and if $(s_{\mm_{s_1,h_1}},h_{\mm_{s_1,h_1}})$
% and $(s_{\mm_{s_2,h_2}},h_{\mm_{s_2,h_2}})$
% are obtained by applying the two transformations above, then
% $(s_{\mm_{s_1,h_1}},h_{\mm_{s_1,h_1}})$
% and $(s_{\mm_{s_2,h_2}},h_{\mm_{s_2,h_2}})$
% are isomorphic.\footnote{$(s_1,h_1)$  and $(s_2,h_2)$ are isomorphic if there exists a bijection
% over $Add$ which changes $(s_1,h_1)$ into $(s_2,h_2)$. }

The above shows how to transform the memory models into each other
in a natural way. More precisely:

\begin{lemma}
If $s,h$ are a stack and a heap,
and if $s_{\mm_{s,h}},h_{\mm_{s,h}}$
are the stack and heap
obtained by applying the above transformations on $s$ and $h$ in order, then
$s_{\mm_{s,h}}=s$ and $h_{\mm_{s,h}}=h$.
\end{lemma}

\begin{remark}
The semantics of the standard programming language used with separation logic allows
memory cells to be reallocated, while our programming language forbids this for technical simplicity. 
\end{remark}

%% file: appendix-nexptime.tex
\begin{table}[t]
  \centering

\[
\begin{array}{ rcl}
  tr_z(C) & = & C(z) \qquad \qquad~ C~\mbox{ is an atomic concept} \\
  tr_{z,\bar{z}}(r) & = & r(z,\bar{z}) \qquad\qquad r~\mbox{ is an atomic role} \\[1ex]
  tr_z(C\sqcap  D) & = &
  tr_z(C)\land  tr_z(D)\\
  tr_z(C\sqcap  D) & = &
  tr_z(C)\lor tr_z(D)\\
  tr_z(\neg C) & = &
  \neg tr_z(C)\\

tr_{z,\bar{z}}(r \sqcap  s) & = &
tr_{z,\bar{z}}(r)\land  tr_{z,\bar{z}}(s)\\

tr_{z,\bar{z}}(r \sqcup  s) & = &
tr_{z,\bar{z}}(r)\lor  tr_{z,\bar{z}}(s)\\

tr_{z,\bar{z}}(r \setminus  s) & = &
tr_{z,\bar{z}}(r)\land \neg tr_{z,\bar{z}}(s)\\

tr_{z,\bar{z}}(r^-) & = &
tr_{\bar{z},z}(r)\\

tr_{z,\bar{z}}(C\times D) & = &
tr_{z}(C)\land tr_{\bar{z}}(D)\\

tr_{z}(\exists r.C) & = &
\exists y. tr_{z,\bar{z}}(r)\land tr_{\bar{z}}(C)\\[1ex]

tr(C\sqsubseteq D) & = & \forall x. tr_x(C)\rightarrow tr_x(D) \\

tr(r\sqsubseteq s) & = & \forall x,y. tr_{x,y}(r)\rightarrow tr_{x,y}(s) \\

tr( \varphi\land  \psi) & = &
tr( \varphi) \land tr(  \psi)\\

tr( \varphi\lor  \psi) & = &
tr( \varphi) \lor tr(  \psi)\\

tr( \neg \varphi) & = &
\neg tr( \varphi)\\

tr(func(r)) & = &
  \forall x\exists^{\leq 1} y. tr_{x,y}(r) \\

tr(\neg  \varphi) & = &
\neg tr( \varphi) \\

\end{array}
\]\caption{Translation of $\L$ into $\mathcal{C}^2$ by employing only two
  variables $x$ and $y$. In each
  translation rule, $z\in \{x,y\}$. Moreover, $\bar{z}=y$ if $z=x$,
  and $\bar{z}=x$ if $z=y$. }
  \label{tab:fo2}
\end{table}

\section{$\ctwotree$-Satisfiability in Memory Structures}\label{app:nexptime}

 In this appendix we prove Lemma \ref{lem:fin-sat-mem}, i.e. we show that satisfiability of 
 $\ctwotree$ formulae by memory structures is in NEXPTIME.
  We employ the fact that \emph{finite
    satisfiability} of $\ctwotree$-formulae, i.e. truth in a structure with a
  finite domain, is in NEXPTIME:
   \begin{theorem}[W. Charatonik and P. Witkowski \cite{pr:CharatonikWitkowski13}]
   Finite satisfiability of $\ctwotree$ is NEXPTIME-complete.  
   \end{theorem}

  To show that satisfiability of a formula $\psi\in \ctwotree$ in a memory
  structure can be decided in non deterministic exponential time, it
  suffices to construct in linear time a formula $\psi_m\in \ctwo$ such that
  $\psi$ is satisfiable in a memory structure iff $\psi_m\land \psi$
  is finitely satisfiable. 
  The formula $\psi_m$ is the conjunction of
  formulae corresponding to requirements (3)-(9) we
  placed on memory structures. 
  The conjoined formulae are the 
  the translations using $tr:\L\to\ctwo$ from Table \ref{tab:fo2} 
  of the following formulae:
  \begin{itemize}[-]
\item $Aux\equiv o_\nil \sqcup o_\true \sqcup o_\false$,
\item $\Addr \sqsubseteq \neg Aux$ and $\Addr \sqcup Aux \equiv \top$,
\item $Alloc \sqsubseteq \neg \MemPool$, $Alloc \sqsubseteq \neg PossibleTargets$,
$MemPool \sqsubseteq \neg PossibleTargets$, and\\
$Alloc \sqcup PossibleTargets \sqcup\MemPool
  \sqsubseteq \Addr$, 
\item $o\sqsubseteq \neg MemPool$ for every constant symbol $o$ in $\psi$,
\item $func(f)$ and $\Addr \sqsubseteq \exists f. \neg MemPool$ for
  every $f$ of $\psi$ with $f\in \tauField$,
\item $\MemPool \sqsubseteq \exists f . o_\nil \sqcup o_\false$ for
  every $f\in \tauField$, and
\item $C\sqsubseteq \neg \MemPool$ for every atomic concept $C\in
  \tau$ with $C\neq \MemPool$.
\end{itemize}
Requirements (1) and (2) hold by the correct choice of vocabulary. 
To see that requirement (10) holds, namely that $Alloc^\mm$ and $PossibleTargets^\mm$ 
are finite while $MemPool^\mm$ is infinite, 
note that any finite model $\mm$ of $\psi_m\land \psi$ is almost the
desired memory structure with $\mm\models \psi$. The desired $\mm$ is
obtained by adding to $\MemPool^\mm$ infinitely many fresh elements
$e$ and setting $f^{\mm}(e)=o_\nil^\mm$ for all $f\in \tauField$.

%% file: appendix-translations.tex
\section{Translations of \texorpdfstring{$\L$}{L} and \texorpdfstring{$\SLls$}{SLls} into \texorpdfstring{$\ctwotree$}{CT2}} \label{app:translations}

As discussed in Section \ref{se:embed}, $\alpha$
translates $\SLls$-formulae into $\L$ almost exactly, with the caveat that
for every list $L$, some redundant cycles may exist in $L$. 
These cycles are not reachable from the variables of the head variable of the list through the $next$ pointer. 
Since $\L$ is a fragment of first order logic, properties related to connectivity  cannot be used to rule out these cycles.
We use $\ctwotree$ to express the necessary connectivity property. 
Recall $\beta^5(\ls)$ stated that the forest $F_1$ coincides with $next$ inside $L$ and
that the forest induced by $\Fone$ on $L$ is a tree.

Here we define the two translations $\alpha$ and $\beta$ so that the satisfy the Lemmas \ref{lem:sl-to-dl-2-alpha}
and \ref{lem:sl-to-dl-2-beta}, restated here:
\begin{lemma}
\begin{enumerate}[i.]
 \item For every heap $h$ and stack $s$, if $\varphi\in\SLls$ then:
if $s,h\models\varphi$ then $\mm_{s,h} \models \alpha(\varphi)$.
\item  For every heap $h$ and stack $s$, if $\varphi\in\SLls$ then:
$s,h\models\varphi$ iff $\mm_{s,h} \models \beta(\varphi)$.
\end{enumerate}
\end{lemma}

It is convenient to define the following notation:
if $\varphi=\Pi\bv\Sigma$, then there exist $\delta_1,\ldots,\delta_r$
such that $\Sigma=\delta_1 * \cdots * \delta_r$
and the $\delta_i$ are of the form
$var_{1}\mapsto[f_{i}:E_{i}]$.
%or $ls(E_1,E_2)$.
We use the concepts $P_1,\ldots,P_r$ which partition
the allocated memory cells according to $\Sigma$.

First we define the formula $\alpha(\varphi)$ for $\varphi\in \SL$. \\
\[
\begin{array}{lll}
\alpha({E_{1}=E_{2}}) & = &
\left(o_{E_{1}}\equiv o_{E_{2}}\right)\\
\alpha(E_{1}\not=E_{2}) & = &
\neg\left(o_{E_{1}}\equiv o_{E_{2}}\right)\\
\alpha(\Pi_{0}\land\Pi_{1})
&= &
\alpha(\Pi_{0})\land\alpha(\Pi_{1})\\
\alpha(true)
&=&
\left(\top\sqsubseteq\top\right) %&
\\
 \alpha(\emp)
&=&
\left(Alloc\equiv\bot\right)\\
\alpha(\delta_t)
&=&
(P_t \equiv o_{var_{1}})
\land\bigwedge_{i=1}^{k}((o_{var_{1}},o_{E_{i}})\sqsubseteq f_{i}),\\
& & \mbox{for }\delta_t = var_{1}\mapsto[f_{i}:E_{i}:i=1,\ldots,k]\\
%\alpha(\delta_t) & = &
% \alpha_{\ls,m}(var_1,var_2,next,P_t), \mbox{ for } \delta_t=\ls(var_1,var_2)\\
\alpha(\Sigma)
&=&
 \bigwedge_{1\leq t\leq r }\alpha(\delta_t)\land \\ 
&&
\bigwedge_{1\leq t_1 < t_2\leq r }( P_{t_1}\sqcap P_{t_2} \equiv \bot)\\
\alpha(\Pi\bv\Sigma)
&=&
P_1\sqcup \cdots \sqcup P_r\equiv Alloc
\land\alpha(\Pi)\land\alpha(\Sigma)
\end{array}
\]
By construction we have 
for every $\varphi\in \SL$, $s,h\models \varphi$  iff $\mm_{s,h} \models \alpha(\varphi)$.

Now we turn to $\SLls$. Here $\Sigma=\delta_1 * \cdots * \delta_r$
where the $\delta_i$ are of either of the forms
$var_{1}\mapsto[f_{i}:E_{i}]$ or $ls(E_1,E_2)$.
If $\delta_i=\ls(E_1,E_2)$, then 
$\alpha(\delta_i)$ is defined similarly to the definition of $\alpha(\ls)$ in
Section \ref{se:embed} using $P_i$, $E_1$ and $E_2$: 
\[
 \begin{array}{llll}
\alpha^1(\delta_i)=(o_{E_1}  \sqsubseteq  P_i)\\
\alpha^2(\delta_i)=(o_{E_2}  \sqsubseteq  \exists next^{-}.P_i)\\
\alpha^3(\delta_i)=(o_{E_2}  \sqsubseteq  \neg P_i)\\
\alpha^4(\delta_i)=(P_i  \sqsubseteq  o_{E_1}\sqcup \exists next^{-}.P_i)\\
\end{array}
\]
Let 
\begin{eqnarray*}
 \alpha_{emp-ls}(\delta_i)&=& (P_i\sqsubseteq \bot) \land (o_{E_1}=o_{E_2})\\
  \alpha(\delta_i) &=& \alpha^1(\delta_i) \land \cdots \land \alpha^4(\delta_i)\lor \alpha_{emp-ls}(\delta_i)
\end{eqnarray*}
and $\alpha(\ls(E_1,E_2)) = \alpha^1(\ls(E_1,E_2)) \land \cdots \land \alpha^4(\ls(E_1,E_2))$.

Now we turn the translation $\beta$. Similarly to $\beta^5(\ls)$ from Section \ref{se:embed}, 
we define $\beta^5(\delta_i)$ for each $\delta_i=\ls(E_1,E_2)$ as:
\[
\begin{array}{l}
\forall x\forall y\, \big[(P_i(x)\land P_i(y))\to (\Fone(x,y) \leftrightarrow next(x,y))\big]\land \\
 \forall x \big[\big(P_i(x)\land\forall y \,(P_i(y)\to\neg \Fone(y,x))\big)\to (x\approx o_{var_1})\big]
\end{array}
\]
$\beta^5(\delta_i)$ states that the forest $F_1$ coincides with $next$ inside $L$ and
that the forest induced by $\Fone$ on $P_i$ is a tree.
Let $I\subseteq \{1,\ldots,r\}$ be the set of $i$
such that $\delta_i$ is of the form $\ls(E_1,E_2)$. Let
\[
\begin{array}{ll}
\beta^5(\Sigma)& = \bigwedge_{i\in I}\beta^5(\delta_i)  \\ 
\beta(\Pi\bv\Sigma) &=\exists \Fone\, tr(\alpha(\Pi\bv\Sigma)) \land \beta^5(\Sigma) \land \varphi_{forest}(F_1)
\end{array}
\]
Note that to get that $\beta^5(\Sigma)$ indeed states the connectivity condition for each of the lists we use the fact that
$P_1\ldots,P_r$ are disjoint, and therefore the trees we quantify for the different lists are disjoint. 

%% file: appendix-cyclic.tex
\section{Cyclic Data Structures in $\ctwotree$} \label{app:cyclic}

Here we want to clarify that cyclic data structures such as cyclic lists which are expressible in $\SLls$
are translated correctly into $\ctwotree$ by $\beta$. 

Consider the formula $ls(a,b)*ls(b,a)$ which defines a cyclic list with at least two elements. 
In the translation to $\ctwotree$, $P_1$ contains the elements of the list from $a$ to $b$, 
and $P_2$ contains the elements of the list from $b$ to $a$. 
Importantly, $b$ does not belong to $P_1$, and $a$ does not belong to $P_2$. 
This is captured by $\alpha^3$ in the translation of $\SLls$ to $\L$ in Section \ref{se:embed}
or Appendix \ref{app:translations}. 

The translation requires (in $\beta^5$) that the forest $F_1$ coincides with $next$ \textbf{inside} $P_1$ and $P_2$. 
However, crucially, there is no requirement on $next$ \textbf{between} $P_1$ and $P_2$, see
the definitions of $\beta^5$ and $\beta$ in Section \ref{se:embed} and Appendix \ref{app:translations}. 

As a result, in the cyclic list $ls(a,b)*ls(b,a)$, not all $next$ edges are 
required to belong to $F_1$. Rather, the two edges that point to $a$ and to $b$ respectively are not required to belong to $F_1$.

In more detail, if $\mm \models ls(a,b)*ls(b,a)$ then:
\begin{enumerate}
 \item $P_1^\mm=\{a_1,...,a_i\}$ such that $a_1=a$ 
 and the edges $(a_j,a_{j+1})$ belong both to $next^\mm$ and to $F_1^\mm$,
  for $1\leq j\leq i-1$.
 \item The edge $(a_i,b)$ belongs to $next^\mm$ but might not belong to $F_1^\mm$.
 \item $P_2^\mm=\{b_1,...,b_k\}$ such that $b_1=b$ and
 the edges $(b_j,b_{j+1})$ belong both to $next^\mm$ and to $F_1^\mm$, for $1\leq j\leq k-1$.
 \item The edge $(b_k,a)$ belongs to $next^\mm$ but might not belong to $F_1^\mm$.
 \end{enumerate}

 Additionally note that the translation of a formula of the form $var_{1}\mapsto[f_{i}:E_{i}]$ 
 also does not require the $next$ edge from $var_1$ to belong to $F_1$. 

%% file: appendix-backwards.tex
\section{Backwards propagation}
\label{se:backwards-app}\label{app:backward}

Here we give the exact formulation of the backwards propagation lemma
and prove it. 

It is convenient to consider a program $\overline{S}$, which behaves like $S$, except that it does not abort. 
$\overline{S}$ uses a fresh variable $abo$ to indicate whether $S$ aborts. 
 The command $abo := \false$ is added at the beginning of the code.
 Every command $C$ of the form $var_1 := var_2.f$, $var_2.f := var_1$ or $dispose(var_2)$ is replaced with $\overline{C} = if\,\, var_2\,\, =\,\, \nil\,\, then\,\, abo := \true\,\, else\,\, C\,\, fi$.
 For $if$, $assume$ and assignments of the form $var_1.f_1 := var_2.f_2$ commands the case is similar, except that there may be two evaluations
of the form $var_i.f_j$, which need to be reflected in the condition in $\overline{C}$.
By the construction of $\overline{S}$, $\overline{S}$ has the following properties:
\begin{enumerate}
 \item The run of $\mm_1$ on $\overline{S}$ does not abort for any $\mm_1$. 
 \item $abo$ has the value $\true$ at the end of the run of $\overline{S}$ on $\mm_1$
if and only if $S$ aborts on $\mm_1$.
 \item If $\left\langle S,\mm_1\right\rangle\leadsto \mm_2$, then $\left\langle \overline{S},\mm_1\right\rangle\leadsto \mm_2$. 
\end{enumerate}

We need a further extension of our structures, which uses a refined $\leadsto$ relation. 
The refined $\leadsto$ relation will get rid of some non-determinism in the semantics of the programming language. 

Given a finite set $Y$ of labels and a tuple $\bar{d}_Y=(d_y: y\in Y)$ of elements of $M$,
we denote by $\left\langle \mm_1, (\bar{R}^{ext})^{\mm_1}, \bar{d}_Y \right\rangle$
the structure obtained from $\left\langle \mm_1, (\bar{R}^{ext})^{\mm_1}\right\rangle$ by adding the constants $d_y$ for each $y\in Y$.
The vocabulary $\tau_{Y}^{ext}$ of $\left\langle \mm_1, (\bar{R}^{ext})^{\mm_1}, \bar{d}_Y \right\rangle$ extends
$\tau^{ext}$ by constant symbols $\{o_y: y\in Y\}$. 

Given a loopless program $S'$, we assign unique labels $y$ to the commands of $S'$.
For any loopless program $S'$, we denote by $Y_{S'}$ the set of labels of commands in $S'$.
The $\leadsto_{{\bar{d}_{Y_S'}}}$ relation is obtained from the $\leadsto$ relation as follows:\\
$\left\langle S',\mm_1\right\rangle\leadsto_{\bar{d}_{Y_{S'}}}  \mm_2$ iff $\left\langle S',\mm_1\right\rangle\leadsto \mm_2$, except in the three following cases
for $S'$:
\begin{itemize}
 \item[--] $y: var_1 := var_2.f$: $\left\langle S',\mm_1\right\rangle\leadsto_{d_y}  \mm_2$ iff 
 $\left\langle S',\mm_1\right\rangle\leadsto  \mm_2$ and  $f^{\mm_1}(o_{var_2}^{\mm_1})=d_y$. 
           Else, $\left\langle S',\mm_1\right\rangle\leadsto_{d_y} \abo$. 
 \item[--] $y: var:=new$: $\left\langle S',\mm_1\right\rangle\leadsto_{d_y}  \mm_2$ iff 
 $\left\langle S',\mm_1\right\rangle\leadsto  \mm_2$ and  $o_{var}^{\mm_2}=d_y$. 
           Otherwise, $\left\langle S',\mm_1\right\rangle\leadsto_{d_y} \abo$. 
 \item[--] $S'_1; S'_2$: If $\left\langle S'_1,\mm_1\right\rangle\leadsto_{\bar{d}_{Y_{S'_1}}}  \mm'$ and
$\left\langle S'_2,\mm'\right\rangle\leadsto_{\bar{d}_{Y_{S'_2}}}  \mm_2$, 
then $\left\langle S'_1;S'_2,\mm_1\right\rangle\leadsto_{\bar{d}_{Y_{S'_1}}\cup\bar{d}_{Y_{S'_2}}}  \mm_2$. 
\end{itemize}
The main observation is:
\begin{lemma}
For any two memory structures $\mm_1$ and $\mm_2$, 
$\left\langle S',\mm_1\right\rangle\leadsto \mm_2$ iff there exists a tuple $\bar{d}_{Y_S'}$ such that
$\left\langle S',\mm_1\right\rangle\leadsto_{\bar{d}_{Y_S'}}  \mm_2$ .
\end{lemma}

We are now ready to state Lemma \ref{lem:backwards--} precisely:
\begin{lemma}\label{lem:backwards-app} \label{se:backwards-ghost-app}
Let $S$ be a loopless program, $Y_{\overline{S}}$ be the set of labels of commands in $\overline{S}$,
$\mm_1$ and $\mm_2$ be  memory structures, $\bar{d}_{Y_{\overline{S}}}$ be a tuple of $M$ elements
labeled with the labels in $Y_{\overline{S}}$, $d_{abo}\in M$, and
 $\varphi$ be an $\L$-formula over $\tau$.
\begin{enumerate}
 \item If $\left\langle S, \mm_1\right\rangle \leadsto_{\bar{d}_{Y_{\overline{S}}}} \mm_2$, 
 then:\\
 $\mm_2\models \varphi$ iff
$\left\langle \mm_1, (\bar{R}^{ext})^{\mm_1}, \bar{d}_{Y_{\overline{S}}}, d_{abo} \right\rangle \models \Theta_{S}(\varphi)$.
\item If $\left\langle S, \mm_1\right\rangle\leadsto_{\bar{d}_{Y_{\overline{S}}}}  \abo$, then
  for every tuple $\bar{d}_{Y_{\overline{S}}}$ of $M$ elements, 
$\left\langle \mm_1, (\bar{R}^{ext})^{\mm_1},
\bar{d}_{Y_{\overline{S}}},d_{abo} \right\rangle \not\models \Theta_{S}(\varphi)$.
\end{enumerate}
The vocabulary of the structure $\left\langle \mm_1, (\bar{R}^{ext})^{\mm_1}, \bar{d}_{Y_{\overline{S}}}, d_{abo} \right\rangle$ is 
$\tau^{ext}_{Y_{\overline{S}}} \cup\{o_{abo}\}$. 

The definition of $\Theta_S$ is:
\begin{definition}$\Theta_{S}(\varphi)=\Phi_{\overline{S}}(\varphi\land (o_{abo} \equiv o_\false)))$,
$\Phi$ is obtained from $\Psi$ by
substituting every  symbol $R\in\tau^{rem}$ in $\varphi$ by $R^{ext}$, and
$\Psi$ is defined as:
\[\begin{array}{lll}
\Psi_{skip}(\varphi) &=&\varphi\\
\Psi_{var_{1}:=e}(\varphi)&=&\varphi[{o_{var_{1}} /  o_{e}}],\, e=var_{2}\mbox{ or }e=\nil\\
\Psi_{y: var_{1}:=var_{2}.f}(\varphi)&=&\varphi[o_{var_{1}} / o_y] \land (\exists f^{-}.o_{var_{2}}\equiv o_y)\\
\Psi_{var_{1}.f:=e}(\varphi)&=&\varphi[f / f\backslash(o_{var_{1}}\times\top)\cup(o_{var_{1}},o_{e})],\, \\
&&\mbox{where }e=var_{2}\mbox{ or }e=\nil\\
\Psi_{if\,\, b\,\, then\,\,  S_{1}\,\, else\,\, S_{2}\,\, fi}
%\substack
(\varphi)&=&\varepsilon_{b}\land\Psi_{S_{1}}(\varphi)\lor\neg\varepsilon_{b}\land\Psi_{S_{2}}(\varphi)\\
\Psi_{y:var:=new}(\varphi)&=&\varphi[o_{var} /  o_{y}][Alloc / Alloc\sqcup o_{y}] \\
&&\land o_y\sqsubseteq \neg Alloc\\
\Psi_{dispose(var)}(\varphi)&=&\Psi_{S_{disp}}(\varphi[Alloc / Alloc\sqcap\neg o_{var}]),\\
&&\mbox{where }S_{disp}=var.f_{k_1} := \nil; \\ 
 && \cdots;  var.f_{k_w} := \nil\\
\Psi_{S_{1};S_{2}}(\varphi)&=&\Psi_{S_{1}}(\Psi_{S_{2}}(\varphi)) \
\end{array}
\]
The notation $\varphi[A / B]$ should be interpreted as
the syntactic replacement of any occurrence of $A$ with $B$.
We write e.g. $y: var:= new$ to indicate that the command $var := new$ is labeled with $y$.
$\varepsilon_{b}$ is defined inductively:
for
$e_{1}=e_{2}$ we set $\varepsilon_{b}=(A_{e_{1}}\equiv A_{e_{2}})$,
with $A_{var}=o_{var}$ and $A_{var.f}=\exists f^{-}.o_{var}$; $\varepsilon$
extends naturally to the Boolean connectives.
In the definition of $\Psi_{dispose(var)}$, $f_{k_1},\ldots,f_{k_w}$ are the members of $\tauField$ which
occur in $\varphi$. W.l.o.g. we assume that $S$ does not contain commands of the form $if\,\,b\,\,then\,\,S_1\,\, fi$
or $var_1.f_1 := var_2.f_2$, since they can be expressed using the other commands.
\end{definition}
\end{lemma}

To prove Lemma \ref{lem:backwards-app} we need the following lemma:

\begin{lemma}\label{lem:backwards2-app}
Let $S$ be a loopless program without $assume$ commands, $Y_S$ be the set of labels of commands in $S$,
$Y$ be a set of labels disjoint from $Y_S$,
$\mm_1$ and $\mm_2$ be memory structures with universe $M$
and $\bar{d}_{Y_S}$ a tuple of $M$ elements
such that 
$\left\langle S, \mm_1\right\rangle\leadsto_{\bar{d}_{Y_S}} \mm_2$.
% and $(\nn_1)_{\gho} = (\nn_2)_{\gho}$,
Let $\bar{d}_Y$ be a tuple of $M$ elements and
 $\varphi$ be an $\L$-formula over $\tau\cup\{o_{y}: y\in Y\}$.
$\left\langle \mm_2, \bar{d}_Y \right\rangle\models \varphi$ iff
$\left\langle \mm_1, (\bar{R}^{ext})^{\mm_1}, \bar{d}_Y, \bar{d}_{Y_S} \right\rangle \models \Phi_S(\varphi)$.
\end{lemma}

\begin{proof}
We prove the lemma by induction.
\begin{itemize}
\item $S=skip$: $Y_S = \emptyset$, and we have $\left\langle \mm_2, \bar{d}_Y \right\rangle\models \varphi$
      iff $\left\langle \mm_1, \bar{R}^{ext}, \bar{d}_Y \right\rangle\models \Phi_S(\varphi)$, as required.
\item $S = if\,b\,then\,S_1\,else\,S_2\, fi$: depending on whether $\varepsilon_b$ is true or false, $\Phi_{S_1}(\varphi)$ or $\Phi_{S_2}(\varphi)$
      should be used.
\item $var_1 := e$, where $e$ is a variable $var_2$ or $\nil$: every reference to $var_1$ in $\varphi$ is replaced
      by a reference to $var_2$ or $\nil$, respectively.
\item $var_1 := var_2.f$: every reference to $var_1$ in $\varphi$ is replaced with a reference to $o_y$, whose interpretation is $d_y$,
      in accordance with $\leadsto_{\bar{d}_{Y_S}}$, which requires that $d_y$ be the result of applying $f$ on $var_2$.
\item $var_1.f := e$, where $e$ is a variable $var_2$ or $\nil$: the function symbol $f$ is updated by
      removing the current value of $f$ on $var_1$ by subtracting $(o_{var_{1}}\times\top)^{\mm_1}$ from $f^{\mm_1}$
      and setting the new value explicitly
      by adding the pair $(o_{var_1}^{\mm_1},o_{e}^{\mm_1})$ to $f^{\mm_1}$.
\item $S= var := new$ with label $y$: $Y_S = \{y\}$ and $\bar{d}_{Y_S}=(d_y)$. By the definition of $\leadsto_{\bar{d}_{Y_S}}$
      for $new$ commands, $\{d_y\} = Alloc^{\mm_2}\backslash Alloc^{\mm_1}$.
      $\Phi_S(\varphi)$ adds $o_y$ to $Alloc$ and replaces
      every reference to $var$ by a reference to $o_y$.

\item $S = dispose(var)$: $\Phi_S(\varphi)$ removes $var$ from $Alloc$, and using an application of $\Phi$
      to the program $var.f_{k_1} := \nil; \cdots; var.f_{k_w} := \nil$, sets all of the fields
      in $\varphi$ to $\nil$.
\item $S=S_1;S_2$: $\Phi_S(\varphi) = \Phi_{S_1}(\Phi_{S_2}(\varphi))$.
      Let $\mm_3$ be an memory structure such that
      $\left\langle S_1,  \mm_1 \right\rangle\leadsto_{\bar{d}_{Y_{S_1}}} \mm_3$ 
      and $\left\langle S_2,  \mm_3 \right\rangle\leadsto_{\bar{d}_{Y_{S_2}}} \mm_2$.
      We have ${\bar{d}_{Y_{S}}} ={\bar{d}_{Y_{S_1}}}\cup {\bar{d}_{Y_{S_2}}}$.  
      %and
      %$(\nn_1)_{\go} = (\nn_3)_{\gho}$. 

      Consider first $\Phi_{S_2}(\varphi)$. By the induction hypothesis, 
      $\left\langle \mm_2, \bar{d}_Y \right\rangle\models \varphi$ iff
      $\left\langle \mm_3,  \bar{R}^{\mm_2}, \bar{d}_Y, \bar{d}_{Y_{S_2}} \right\rangle \models \Phi_{S_2}(\varphi)$. 

      Let $\mm_4$ be obtained from $\mm_3$ be replacing every relation $R^{\mm_3}$ with $R^{\mm_2}$ for $R\in \tau^{rem}$.
      We have 
      $\Psi_{S_2}(\varphi) \models \left\langle \mm_4,\bar{d}_Y, \bar{d}_{Y_{S_2}} \right\rangle$ 
      iff $ \Phi_{S_2}(\varphi) \models \left\langle \mm_3,  \bar{R}^{\mm_2}, \bar{d}_Y, \bar{d}_{Y_{S_2}} \right\rangle$.

      Since we have $\left\langle S_2,  \mm_4 \right\rangle\leadsto_{\bar{d}_{Y_{S_2}}} \mm_2$, 
      we can apply the induction hypothesis once again, this time on $\Psi_{S_2}$.
      We get that 
      $\left\langle \mm_4, \bar{d}_Y, \bar{d}_{Y_{S_2}} \right\rangle\models \varphi$ iff
      $\left\langle \mm_1,  \bar{d}_Y, \bar{d}_{Y_{S_2}}, \bar{d}_{Y_{S_1}} \right\rangle \models \Phi_{S_1}(\Phi_{S_2}(\varphi))$.
      Hence, 
      $\left\langle \mm_2, \bar{d}_Y \right\rangle\models \varphi$ iff
      $\left\langle \mm_1, \bar{R}^{\mm_2}, \bar{d}_Y, \bar{d}_{Y_{S}} \right\rangle \models \Phi_{S_1}(\Phi_{S_2}(\varphi))$.

\end{itemize}

\end{proof}

\begin{proof}[Proof of Lemma \ref{lem:backwards-app}]

Using Lemma \ref{lem:backwards2-app} with $Y=\emptyset$,
if $\left\langle S, \mm_1\right\rangle\leadsto_{\bar{d}_{Y_S}} \abo$, then
for every tuple of relations $\bar{U}$ interpreting $\bar{R}^{ext}$ we have
$\left\langle \mm_1, \bar{U}, \bar{d}_{Y_S}, d_{abo} \right\rangle \not\models \Phi_{\overline{S}}(\varphi\land (o_{abo} \equiv o_\false))$,
because $abo$ is set to true during the run of $\overline{S}$. 
If $\left\langle S, \mm_1\right\rangle\leadsto_{\bar{d}_{Y_S}} \mm_2$,
% and 
%$(\nn_1)_{\gho} = (\nn_2)_{\gho}$, then there exists $\bar{d}_{Y_S}$ and $d_{abo}$ such that
$\mm_2 \models \varphi$  iff
$\left\langle \mm_1, \bar{d}_{Y_S}, d_{abo} \right\rangle \models \Phi_{\overline{S}}(\varphi\land (o_{abo} \equiv o_\false))$.

Note that $d_{abo}$ is the value of $abo$ at the beginning of the run of $\overline{S}$.
Since the first command of $\overline{S}$ assigns $abo$ a new value, $d_{abo}$ plays no role (it appears because, technically, $abo$ still needs a value at the beginning of the run).

Also note that in Lemma \ref{lem:backwards--}, $\bar{d}_{Y_{\overline{S}}}$ strictly extends $\bar{d}_{Y_S}$, since
$\overline{S}$ extends $S$. However, the semantics of all of the new commands in $\overline{S}$ does not actually
depend on the relevant $d_y$ (since none of them of $new$ commands or assignments of the form $var_1.var_2.f$). 
Hence, any extension of  $\bar{d}_{Y_S}$ into $\bar{d}_{Y_{\overline{S}}}$ will do. 
\end{proof}

\begin{remark}
Revisiting the example in Section \ref{se:backwards-example} with the detailed version of Lemma \ref{lem:backwards-app} in mind,
we now see that $\Theta$ is actually the backwards propagation of programs of the form $\overline{S}$. 
The backward propagation of $\overline{\lambda(\ell_l.\ell_l)}$ is similar to that presented in Section
\ref{se:backwards-example}, with 
$\Theta_{\overline{\lambda(\ell_l.\ell_l)}}(\varphi_{p-as-\ell_l}) = \Phi_{\overline{\lambda(l.l)}}(\varphi_{p-as-\ell_l}\land
(o_{abo} \equiv o_\false))$. 

\end{remark}

%% file: appendix-soundness.tex
\section{Soundness and Completeness Theorems}\label{app:soundness-completeness}
Here we restate and prove Theorems \ref{th:soundness-completeness} and \ref{th:soundness-2}:
{
\renewcommand{\thetheorem}{\ref{th:soundness-completeness}}
\begin{theorem} \label{th:soundness-completeness-app}
 Let $G$ be a program such that $shp$ is inductive for $Init$ and $Init\models cnt(\ell_{init})$. 
The following statements are equivalent:
 \begin{enumerate}[(i)]
  \item  For all $e\in E$, $VC(e)$ holds.
  \item $\dmnCnt$ is inductive for $Init$ relative to $shp$.
 \end{enumerate}
 \end{theorem}
}
\begin{proof}
 Assume $VC(e)$ holds for every $e\in E$. Let 
 $e=(\ell_1,\ell_2)\in E$ and memory structures 
 $\mm_1$ and $\mm_2$ such that $\mm_1\models g(\ell_1)$ and 
 $\left\langle \lambda(e),\mm_1\right\rangle \leadsto \mm_2$. 
 Since $\mm_1\models g(\ell_1)$ we have 
 $s_{\mm_1},h_{\mm_1}\models shp(\ell_1)$ and 
 $\mm_1\models g(\ell_1)$. Since $shp$ is inductive, $s_{\mm_2},h_{\mm_2} \models shp(\ell_2)$. 
 There exists a tuple $\bar{d}_{Y_{\lambda(e)}}$ such that 
 $\left\langle \lambda(e),\mm_1\right\rangle \leadsto_{\bar{d}_{Y_{\lambda(e)}}} \mm_2$. 

Let $(\bar{R}^{ext})^{\mm_1}$ be the tuple
of copies of $R^{\mm_2}$ relations from Section \ref{se:backwards--}, i.e. $(\bar{R}^{ext})^{\mm_1}$ is
$
 \big((R^{ext})^{\mm_1} : (R^{ext})^{\mm_1}=R^{\mm_2}\mbox{ and }R\in\tau\backslash\tauField\big)
$.
Let $\nn = \left\langle \mm_1,(\bar{R}^{ext})^{\mm_1},\bar{d}_{Y_{\lambda(e)}}\right\rangle$.  
Since $s_{\mm_1},h_{\mm_1}\models shp(\ell_1)$ and $\mm_1\models \dmnCnt(\ell_1)$,
$\nn \models \beta(shp(\ell_1))\land tr(cnt(\ell_1))$. 

By Lemma \ref{se:backwards-ghost-app},
 $\nn\models
 \Theta_{\lambda(e)}\big(\alpha(shp(\ell_2)) \land \neg  \dmnCnt(\ell_2)\big)$
 iff $\mm\models \alpha(shp(\ell_2)) \land \neg  \dmnCnt(\ell_2)$. 
 Since 
 $\nn\models VC(e)$
 and $\nn\models \beta(shp(\ell_1))\land tr(cnt(\ell_1))$,
 it must be that 
 $\nn\not\models 
 \Theta_{\lambda(e)}\big(\alpha(shp(\ell_2)) \land \neg  \dmnCnt(\ell_2)\big)$,
 so $\mm\not \models \alpha(shp(\ell_2)) \land \neg  \dmnCnt(\ell_2)$. 
 Since $\mm\models \beta(shp(\ell_2))$, in particular 
 $\mm \models \alpha(shp(\ell_2))$. Hence
 $\mm\models \dmnCnt(\ell_2)$. 
 We get that $\dmnCnt$ is inductive for $Init$ relative to $shp$.

 Conversely, assume $\dmnCnt$ is inductive for $Init$ relative to $shp$.
 Assume for contradiction that there exists $e=(\ell_1,\ell_2)\in E$ such that $VC(e)$ does not hold
 Then there exists a memory structure $\nn$ such that 
  \begin{eqnarray*}
    \nn &\models &  \beta(shp(\ell_1))\land tr(\dmnCnt(\ell_1)) \\ 
     & &  \land tr\left( \Theta_{\lambda(e)}\big(\alpha(shp(\ell_2)) \land \neg  \dmnCnt(\ell_2)\big)\right)
 \end{eqnarray*}
 Let $\nn  = \left\langle \mm_1,(\bar{R}^{ext})^{\mm_1},\bar{d}_{Y_{\lambda(e)}},d_{abo}\right\rangle$.
 Then $\mm_1$ is also a memory structure and 
 $\mm_1\models \beta(shp(\ell_1))\land tr(\dmnCnt(\ell_1))$,
 so $s_{\mm_1},h_{\mm_1}\models shp(\ell_1)$ and $\mm_1\models \dmnCnt(\ell_1)$. 
 Since $\nn \models \Theta_{\lambda(e)}(\cdots)$, $abo$ must not be set to true
 in the computation of $\lambda(e)$ starting from $\mm_1$ with $\bar{d}_{Y_{\lambda(e)}}$.
 Therefore, there exists $\mm_2$ such that 
 $\left\langle\lambda(e),\mm_1\right\rangle \leadsto_{\bar{d}_{Y_{\lambda(e)}}} \mm_2$. 
 Since the computation of $\lambda(e)$ is not affected by the interpretations of
 $R\in \tau\backslash\tauField$, assume w.l.o.g. that for each $R\in \tau\backslash\tauField$, 
 $R^{\mm}=(R^{ext})^{\mm_1}$. 
 Since $\left\langle\lambda(e),\mm_1\right\rangle \leadsto \mm_2$
 and $s_{\mm_1},h_{\mm_1}\models shp(\ell_1)$, we get $s_{\mm_2},h_{\mm_2}\models shp(\ell_2)$.  
 By Lemma \ref{lem:backwards-app}, $\mm_2 \models \alpha(shp(\ell_2)) \land \neg  \dmnCnt(\ell_2)$. 
 In particular, $\mm_2 \not\models  \dmnCnt(\ell_2)$, in contradiction to $\dmnCnt$ 
 being inductive relative to $shp$. 
 
 \end{proof}

%--------------------------------------------------------------------------------
For every $e=(\ell_1,\ell_2)\in E$, let $Reach(e)$ be the set of memory structures 
$\mm\in Reach(\ell_2)$
for which there exist
 $\mm_{init} \in \Init$ and a path $P$ in $G$ starting at $\linit$ and ending with $e$
 such that $\left\langle P,\mm_{init}\right\rangle\leadsto^* \mm$.

{
\renewcommand{\thetheorem}{\ref{th:soundness-2}}
\begin{theorem}\label{th:soundness-2-app}
Let $G$ be a program such that $shp$ is inductive for $Init$ and $Init\models cnt(\ell_{init})$. 
 If for all $e\in E$,
 $VC(e)$ holds, then for $\ell\in V$, $Reach(\ell)\models \dmnCnt(\ell)$.
 \end{theorem}
}
 
 \begin{proof}
By Theorem \ref{th:soundness-completeness-app}, $\dmnCnt$ is inductive for $Init$ relative to $shp$.  
 Let $\mm \in Reach(\ell)$, and let $P$ be a path and $\mm_{init}\in \Init$ as guaranteed for members of $Reach(\ell)$.
We prove the following claim by induction on the length of $P$:
\begin{quote}
 If for all $e\in E$,
 $VC(e)$ holds, then for $\ell\in V$, $Reach(\ell)\models \dmnCnt(\ell)$
 and $Reach(\ell)\models shp(\ell)$. 
 
\end{quote}

If $P$ is empty, then $\mm\in Init$ and the claim holds.

If $P$ is not empty, let $e=(\ell_0,\ell)$ be the last edge of $P$, and let
$P_0$ be the path obtained from $P$ by removing $e$.
Let $\mm_0$ be a memory structure such that
 $\left\langle  P_0, \mm_{init} \right\rangle \leadsto^* \mm_0$
 and $\left\langle  \lambda(e), \mm_0 \right\rangle \leadsto \mm$.
 Let $\bar{d}_{Y_{\lambda(e)}}$ be a tuple of $M$ elements such that 
 $\left\langle  \lambda(e), \mm_0 \right\rangle \leadsto_{\bar{d}_{Y_{\lambda(e)}}} \mm$.
By the induction hypothesis, $s_{\mm_0},h_{\mm_0}\models shp(\ell_0)$ and
$\mm_0\models \dmnCnt(\ell_0)$.
Since $shp$ is inductive, and since $\dmnCnt$ is inductive relative to $shp$,
$s_{\mm},h_{\mm}\models shp(\ell)$.
and 
$\mm\models \dmnCnt(\ell)$.

\end{proof}

\begin{remark}
In the proofs in this appendix only case 1. of Lemma \ref{lem:backwards-app} is used.
The purpose of case 2. of Lemma \ref{lem:backwards-app} is to make the verification conditions less strict, in
the sense that they require nothing of aborted executions.
\end{remark}